\documentclass[11pt]{article}

\usepackage{amsmath, amsfonts, amssymb}
\usepackage{mathrsfs}
\usepackage{theorem}
\usepackage{bm}
\usepackage{graphicx}
\newtheorem{theorem}{Theorem}[section]

\newtheorem{thm}{Theorem}[section]

\newtheorem{cor}[thm]{Corollary}
\newtheorem{defn}[thm]{Definition}
\newtheorem{remark}{Remark}[section]
\newtheorem{lemma}{Lemma}[section]

{\theorembodyfont{\rmfamily} \newtheorem{rem}[thm]{Remark}}
{\theorembodyfont{\rmfamily} }

\def\e{\text{\rm{e}}}

\newcommand{\fin}{\hspace*{\fill}\rule{0.3em}{1ex}}
\newenvironment{proof}{{\bf \noindent Proof.}}{\fin}

\numberwithin{equation}{section}

\allowdisplaybreaks

\begin{document}

\title{ Foreign exchange market modelling and an on-line portfolio selection algorithm}

\author{Panpan Ren and Jiang-Lun Wu\footnote{Corresponding author.} \\[0.2cm] 
{\small Department of Mathematics, Swansea University, Swansea, UK}\\
{\small  Email: 673788@swansea.ac.uk; j.l.wu@swansea.ac.uk }}
\maketitle

\begin{abstract}
In this paper, we introduce a matrix-valued time series model for foreign exchange market. We then formulate trading matrices, foreign exchange 
options and return options (matrices), as well as on-line portfolio strategies. Moreover, we attempt to predict returns of portfolios by developing 
a cross rate method. This leads us to construct an on-line portfolio selection algorithm for this model. At the end, we prove the profitability and 
the universality of our algorithm. 
\end{abstract}

\noindent \textbf{Keywords}: Foreign exchange market modelling, on-line portfolio, optimisation, cross rate, currency exchange market, 
matrix algorithm.

\section{Introduction}

The problem of constructing on-line portfolio selection scheme for time series financial models (with or without transaction costs) has been discussed in 
the last two decades, see e.g. \cite{Cover,Duffie,CoverOrd,cox0,ALZ2001} (and references therein).  The main effort of this kind study is to develop a universal 
portfolio to optimise the on-line portfolio management asymptotically, which can be traced back to the celebrated problem of Merton (\cite{Merton}, see 
also \cite{Davis,ShreveSoner}).    

Analysing currency trading, or to be more precise, modelling, predicting and hedging foreign currency exchange rates are important problems, giving that 
modern communication technology nowadays brings unified (global) financial market setting worldwide, making the currency exchange markets more and 
more complicated on one side and highly demanding deep mathematical analysis and high computing in the micro level model on the other side. This then 
challenges theoretical considerations as well as computing technology profoundly. The latter is linked to deep machine learning and data analysis. 

In this paper, we use matrix-valued time series to model foreign exchange markets. We aim to establish an on-line  (universal) portfolio selection and to 
analyse the possible optimal algorithm. In our forthcoming work, we plan to exam our algorithm with the help of deep machine learning skills.      

We start with the mathematical settings of the foreign exchange markets.  We use a square matrix-valued time series to represent the daily pair 
wise currency exchange rates. Each matrix for a fixed date is interpreted as follows. We list all currencies in a row and in a column with the same order,  the 
diagonal entries are just each currency against to itself, the upper triangular part stands for the investors to buy any selected currency against other currencies  
and the low triangular part denotes the investors to sell any selected currency against other currencies.  From this set-up, one can formulate trading matrices 
for adjacent days accordingly. We then define foreign exchange options and return matrices (i.e., return options).  Based on these, on-line portfolio strategies  
and transaction costs can be  implemented. 

To avoid high transaction costs incurred during the trading, one has to optimise the associated portfolios. To this end, we introduce a {\it distance} for 
matrix-valued on-line portfolio strategies. and we further formulate properly the returns of portfolios based on return matrices. On the other hand, in order 
to predict the returns of portfolios, we define an {\it order} for return matrices from which we develop a {\it cross rate method}. Finally, we are able to establish 
our cross rate schemes and show the universality of our algorithm. The feature of our consideration is that investors can measure their on-line portfolios by 
applying universality of the two specific update rules with the cross rate method. It allows investors obtain more than half probability chance to get more profit 
for their portfolios. Our work is inspired by \cite{Davis,helmbold,ALZ2001}, however, we are dealing with matrix-valued time series for the foreign exchange markets 
while those papers only treated vector-valued time series with only considering two scaler states for a complete comparison.    

The paper is organized as follows. In the next section, we set up our mathematical framework for the foreign exchange markets. In Section 3, we give a 
full analysis of update rules for on-line portfolio selections. Section 4 is devoted to developing the cross rate approach for the prediction of the returns. 
We present and prove our main results on the profitability and the universality in Section 5. At the last section,  Section 6, we draw our conclusions.        

\section{Preliminaries and mathematical framework}

\subsection{The foreign exchange markets}

To begin with, let us recall some basic background on the foreign exchange markets.There are 164 circulating official currencies around the world. In the foreign 
exchange markets, like XE, there are 39 different currencies listed, which can be traded. Each tradable currency has their own price against to another 
currency, which is called the currency exchange rate. The foreign exchange markets are international decentralised financial markets for trading currencies, so the 
participants can buy, sell and exchange currencies at spot or determined currency exchange rates. In the foreign exchange markets, currencies are regarded as the 
underlying assets. Thus, the price of the underlying asset is the currency exchange rates. As currencies are always traded in pairs, currency exchange rates are 
written as  each termed currency against to the based currency.  For instance, \pounds 0.7/\textdollar1 means that 1 US Dollar (for short USD) can be exchangeable 
to 0.7 British Pound (for short GBP), where ``\textdollar" is the base currency and $\pounds $ is the quote currency (counter currency). So we say the currency exchange 
rate for GBP against to USD is $ 0.7$. Conversely, the currency exchange rate for USD against to the GBP is 1\textdollar/$\pounds 0.7$=1.429. On the basis of the 
currency exchange rates above, the investor, who holds $\pounds 100$ and  \textdollar 100 at the same  time, can sell $\pounds 100$  to get \textdollar142.9, while  
sell  \textdollar100 to get $\pounds 70$. As a matter of fact, the relationship between currency exchange rates for buying and selling a currency is inverse proportion. 

Let us take the pair of GBP and USD as an example to explicate  the  mechanism of the foreign exchange markets.
\begin{figure}[htb]
	\begin{minipage}[b]{1.0\linewidth}
		\centering
		\centerline{\includegraphics[width=1.0\linewidth]{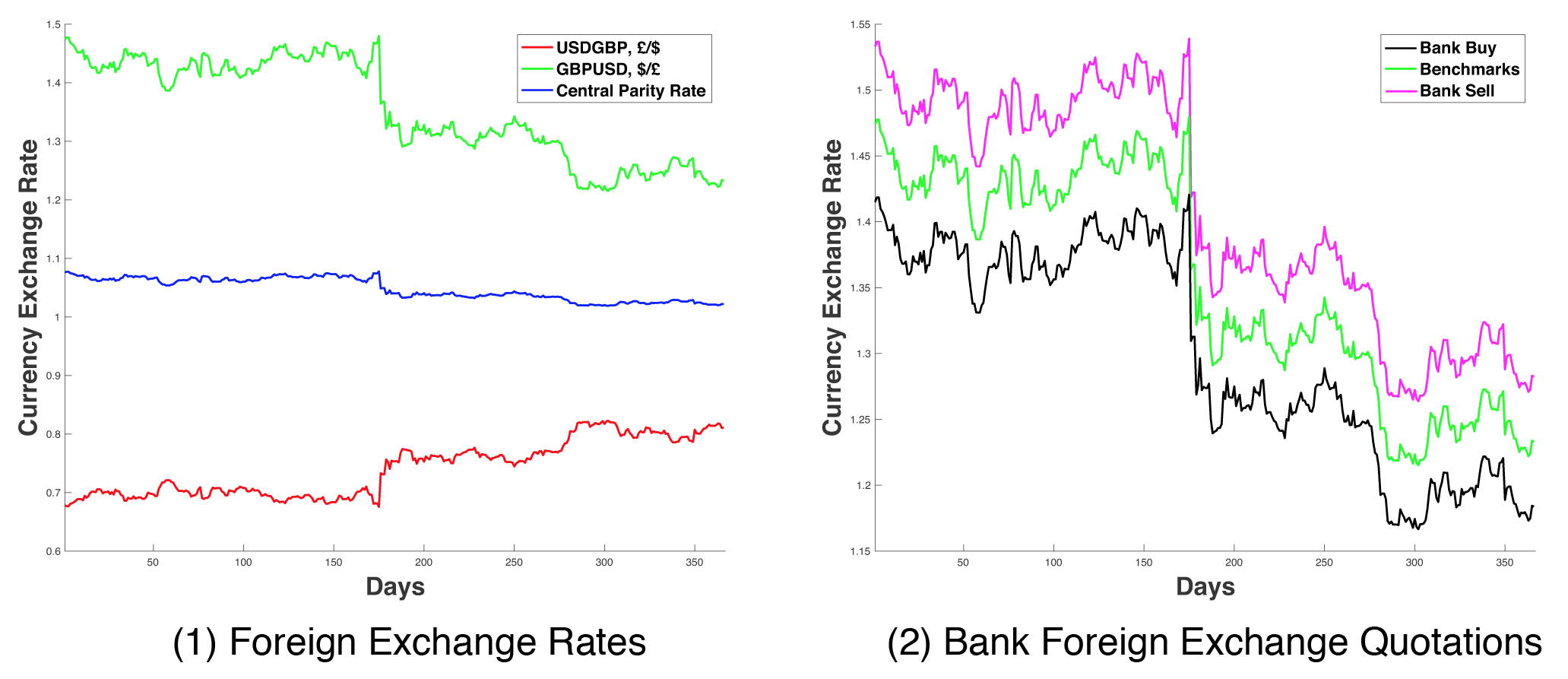}}
	\end{minipage}
	\caption{Currecy Exchange Rates}
	\label{WOZUIMEI}
\end{figure}

 The diagrams below show the currency exchange rates for GBPs against to USDs in 2016. We suppose that the bank is based on UK. In what follows, we shall go into 
 details Figure (1) and Figure (2), one-by-one. 

With regard to the  Figure (1),
\begin{itemize}
	\item The horizontal axis is the time parameter from   day 1 to  day 365
and 	the vertical axis stands for the currency exchange rates for GBPs and USDs;
	
	\item The green line represents the currency exchange rates for GBPs against to USDs, which is the selling price of the USDs for the bank in UK;
	
	\item The red line signifies the currency exchange rates for USDs against to GBPs, which is the buying price of the USDs for the bank in UK;
	
	\item The blue line means the central parity  rates for USDs against to GBPs, which is equal to half of  buying price (see green line) plus selling price (see red line).
	\end{itemize}
		
	As the volatility of currency exchange rates,   banks  also take  risks for the price volatility when they sell and buy the currencies. So,  each bank adopts different currency exchange rates according to the benchmark rates. Meanwhile,  the trader will be charged  for each currency transaction.
	
	Concerning the Figure (2), 
	\begin{itemize}	
	\item  The green line is taken from the Figure (1), which is the benchmark rate  for the bank. According to  the mechanism of foreign exchange markets, banks adjust their own currency exchange rates to hedge the risk from price volatilities. On the other hand, the gap between the pink  and the black line is the profit that   the bank   obtains.
	
	\item  The pink line is the bank quotation of selling foreign currencies (i.e., USDs), which is greater than the benchmark rate at the same moment. 
	
	\item  The black line stands for the bank quotation of buying  foreign currencies (i.e., USDs), which is less than    the benchmark rate at the same moment.	
\end{itemize}

As far as the  investors are concerned,   it is sufficient to look at the pink line (in the Figure (2))  when they buy foreign currencies and watch the black line in the case that they sell foreign currencies. Thus, if an  investor buys foreign currencies at one moment, then the profit for the investor depends heavily on the bank quotation of selling foreign currencies exchange rates (see the pink line in the Figure (2)) when the investor sells the currencies bought at the next moment. 

As we know, there are numerous factors which have impacts on the currency exchange rates (see e.g. \cite{PanpanRenWu} and references therein). Nevertheless, in the present paper, we will not go into 
details the corresponding influential factors.  In this work, we are interested in  the prediction of the returns and  whether there is an algorithm to optimise the currency portfolios.

\subsection{Our foreign exchange market setup}
Before we present the set up of the foreign exchange markets, let us introduce some notation  and terminology. Let $N>1$ be a fixed integer, which will be called the period of investment in the foreign exchange markets. Throughout the paper, we assume that there are $m~ (>1)$ different currencies which are tradable in the foreign exchange markets. Set the (ordered) list of currencies by $\Lambda_0:=\{1,...,m\}$ and 
the trading dates by $\Lambda_1:=\{1,...,N\}$.

 \begin{itemize}
    \item For  each pair  $(i,j)\in\Lambda_0\times\Lambda_0$ with $i<j$,   we use $\bar s_{ij}^{(k)} $  to denote the currency exchange rate for the investor to buy the  foreign currency $j$ via the (home) currency $i$ at the date $k\in\Lambda_1$; and  
    for  $i>j$, $\underline s_{ij}^{(k)}$ stands for the currency exchange rate for the investor to sell the foreign currency $j$ to get the home currency $i$ at the date $k\in\Lambda_1$. 
    
 \item  Note that  there is a corresponding currency exchange rate for each pair of the  $m$ different currencies (indeed, the total number of  pairs are $m^2$ with 
 $m(m-1)$ distinct pairings). Thus,  the  currency exchange rates at the  date $k$ and the adjacent date  $k+1$ can be formed by two square matrices, denoted by $S^{(k)}$ and $S^{(k+1)}$, respectively,  
 as in the following manner 
\begin{equation}\label{SK}
S^{(k)}= \left(\begin{array}{ccccc}
s_{11}^{(k)} & \bar s_{12}^{(k)}&\cdots&\bar s_{1m}^{(k)}\\
\underline s_{21}^{(k)} & s_{22}^{(k)}&\cdots&\bar s_{2m}^{(k)}\\
\vdots & \vdots &\ddots&\vdots\\
\underline s_{m1}^{(k)} &\underline s_{m2}^{(k)}&\cdots&s_{mm}^{(k)}\\
\end{array}
\right) 
\end{equation}
and 
\begin{equation}\label{SK+1}
S^{(k+1)}= \left(\begin{array}{ccccc}
s_{11}^{(k+1)} & \bar s_{12}^{(k+1)}&\cdots&\bar s_{1m}^{(k+1)}\\
\underline s_{21}^{(k+1)} & s_{22}^{(k+1)}&\cdots&\bar s_{2m}^{(k+1)}\\
\vdots & \vdots &\ddots&\vdots\\
\underline s_{m1}^{(k+1)} &\underline s_{m2}^{(k+1)}&\cdots&s_{mm}^{(k+1)}\\
\end{array}
\right).
\end{equation}

\item For a square matrix, the entries above the main diagonal are called the upper triangular part and the entries below the main diagonal are named as the lower triangular part. The upper triangular parts  and the lower triangular parts of $S^{(k)}$ and $S^{(k+1)}$, defined in \eqref{SK}  and \eqref{SK+1}, respectively, are the bank quotations of selling and  buying foreign currencies, respectively. 

\item For $j=i$, $s_{ij}^{(.)}$ means the currency exchange rate of the currency $i$ against to itself at the day $k\in\Lambda_1.$ For simplicity, we set $s_{ij}^{(.)}=1$. 

\item From the point of view for the bank, it is plausible to assume $\bar s_{ij}^{(.)}> \underline s_{ij}^{(.)}> 0$. So, there is a gap between $\bar s_{ij}^{(.)}$ and $ \underline s_{ij}^{(.)}$ so that  there exists $\varepsilon>0$ satisfying $\bar s_{ij}^{(.)}- \underline s_{ij}^{(.)}=2\varepsilon$. Observe that   the parameter $\varepsilon$ will be changed according to  the transaction amount in the spot market.
    
 \item According to the foreign exchange market mechanism (showed in Figure (2)), when trading currencies,  investors buy (or sell) foreign currencies at the date $k$ and then  sell (or buy) foreign currencies at the day $k+1$. Then we have following two matrices for investors to trade currencies at different time.
	\begin{equation*} 
	S_{(k+1)}^{(k)}= \left(\begin{array}{ccccc}
	s_{11}^{(k)} & \bar s_{12}^{(k)}&\cdots&\bar s_{1m}^{(k)}\\
	\underline s_{21}^{(k+1)} & s_{22}^{(k)}&\cdots&\bar s_{2m}^{(k)}\\
	\vdots & \vdots &\ddots&\vdots\\
	\underline s_{m1}^{(k+1)} &\underline s_{m2}^{(k+1)}&\cdots&s_{mm}^{(k)}\\
	\end{array}
	\right),
	\end{equation*}
	and 
	\begin{equation*} 
	S_{(k)}^{(k+1)}= \left(\begin{array}{ccccc}
	s_{11}^{(k+1)} & \bar s_{12}^{(k+1)}&\cdots&\bar s_{1m}^{(k+1)}\\
	\underline s_{21}^{(k)} & s_{22}^{(k)}&\cdots&\bar s_{2m}^{(k+1)}\\
	\vdots & \vdots &\ddots&\vdots\\
	\underline s_{m1}^{(k)} &\underline s_{m2}^{(k)}&\cdots&s_{mm}^{(k)}\\
	\end{array}
	\right),
	\end{equation*}
where the lower triangular parts of  $S_{(k+1)}^{(k)}$  (resp. $S_{(k)}^{(k+1)} $) are the currency exchange rates for investors to sell (resp. buy ) foreign currencies at the date $k$ (resp. the date $k+1$), and then the upper triangular parts of  $S_{(k+1)}^{(k)}$  (resp. $S_{(k)}^{(k+1)} $) are the currency exchange rates for investors to buy foreign currencies at the date $k+1$ (resp. the date $k$).
	
	\item Assume that an investor holds the home currency $i$, where the total value is $X_{i}$. At the day $k$, the investor buys the foreign currency $j$, whose value is equal to $X_{i}\bar s_{ij}^{(k)}$, via the home currency $i$, then we have the value of foreign currency is $X_{j}(=X_{i}\bar s_{ij}^{(k)})$. At the day $k+1$, the investor sells the foreign currency $X_{j}$ to get the home currency $\tilde X_{i}$.
	
	\item Next, let us define two foreign exchange options $\hat{\bar s}_{ij}^{(k+1)}$ and $\hat{\underline s}_{ij}^{(k+1)}$  for the currency exchange rates in the date $k+1$, respectively,  by 
	
\begin{equation}\label{ce1}
\hat{\bar s}_{ij}^{(k+1)}=
\begin{cases}
\underline s_{ij}^{(k+1)} , ~~~~~~~~\underline s_{ij}^{(k+1)}>\bar s_{ij}^{(k)}\\
0 , ~~~~~~~~~~~~~~~~~~\mbox{others}
\end{cases}
\end{equation}

\begin{equation}\label{ce2}
\hat{\underline s}_{ij}^{(k+1)}=
\begin{cases}
\underline s_{ij}^{(k+1)} , ~~~~~~~~\bar s_{ij}^{(k+1)}>\underline s_{ij}^{(k)}\\
0 , ~~~~~~~~~~~~~~~~~~\mbox{others}
\end{cases}
\end{equation}
That is, there is no investment for investors if it is not profitable.

	\item For ${(i,j)}^{th}$ currency pair in the market at the date $k$, the return at the  date $k$ and the date $k+1$ can be expressed respectively as $\bar r_{ij}^{(k)}$, $\underline r_{ij}^{(k)}$ and 
	$\bar r_{ij}^{(k+1)}$, $\underline r_{ij}^{(k+1)}$. Let $\bar s_{ij}^{(ko)}$ and $\underline s_{ij}^{(ko)}$ be the opening exchange rates at the date $k$, and denote  the closing exchange rates 
	of ${(i,j)}^{th}$ currency pair by $\bar s_{ij}^{(kc)}$ and $\underline s_{ij}^{(kc)}$, then we have the following 
	\begin{equation}\label{r1}
     \bar r_{ij}^{(k)}=
     \begin{cases}
\frac{\bar s_{ij}^{(ko)}}{\underline s_{ij}^{(kc)}} , ~~~~~~~~\bar s_{ij}^{(ko)}>\underline s_{ij}^{(kc)}\\
0 , ~~~~~~~~~~~~~~~~~~\mbox{others}; 
\end{cases}  \quad      
     \bar r_{ij}^{(k+1)}=
     \begin{cases}
\frac{\bar s_{ij}^{(ko)}}{\underline s_{ij}^{(k+1)c}} , ~~~~~~~~\bar s_{ij}^{(ko)}>\underline s_{ij}^{(k+1)c}\\
0 , ~~~~~~~~~~~~~~~~~~\mbox{others}
\end{cases} 
    \end{equation}
    and
	\begin{equation}\label{r2}
     \underline r_{ij}^{(k)}=
      \begin{cases}
\frac{\underline s_{ij}^{(ko)}}{\bar s_{ij}^{(kc)}} , ~~~~~~~~\underline s_{ij}^{(ko)}>\bar s_{ij}^{(kc)}\,\\ 
0 , ~~~~~~~~~~~~~~~~~~\mbox{others};
\end{cases} \quad      
     \underline r_{ij}^{(k+1)}=
     \begin{cases}
\frac{\underline s_{ij}^{(ko)}}{\bar s_{ij}^{(k+1)c}} , ~~~~~~~~\underline s_{ij}^{(ko)}>\bar s_{ij}^{(k+1)c}\\
0 , ~~~~~~~~~~~~~~~~~~\mbox{others}. 
\end{cases}     
    \end{equation}
    
	\item The return matrices $R^{(k)}$ and $R^{(k+1)}$ of the $m$ currency exchange transactions at the  trading dates $k$ and $k+1$ are defined respectively by 
	\begin{equation}\label{RK}
	R^{(k)}=  \left(\begin{array}{ccccc}
	r_{11}^{(k)} & \bar r_{12}^{(k)}&\cdots&\bar r_{1m}^{(k)}\\
	\underline r_{21}^{(k)} & r_{22}^{(k)}&\cdots&\bar r_{2m}^{(k)}\\
	\vdots & \vdots &\ddots&\vdots\\
	\underline r_{m1}^{(k)} &\underline r_{m2}^{(k)}&\cdots&r_{mm}^{(k)}\\
	\end{array}
	\right) 
	\end{equation}
	and 
	\begin{equation}\label{RK+1}
	R^{(k+1)}= \left(\begin{array}{ccccc} 
	r_{11}^{(k)} & \bar r_{12}^{(k+1)}&\cdots&\bar r_{1m}^{(k+1)}\\
	\underline r_{21}^{(k+1)} & r_{22}^{(k)}&\cdots&\bar r_{2m}^{(k+1)}\\
	\vdots & \vdots &\ddots&\vdots\\
	\underline r_{m1}^{(k+1)} &\underline r_{m2}^{(k+1)}&\cdots&r_{mm}^{(k+1)}\\
	\end{array}
	\right).
	\end{equation}
Note that $\bar r_{ij}^{(.)} \neq0$ whenever $\underline r_{ji}^{(.)} =0$ and reciprocally $\underline r_{ij}^{(.)} \neq0$ whenever $\bar r_{ji}^{(.)} =0$.  Moreover, 
	 $r_{ji}^{(.)}=0$ for $i=j$.  In literature, the return matrix $R^{(k)}$ is also referred to as a price relative, see, for instance,    \cite{ALZ2001}.
	
	\item For each pair $(i,j)$ of the $m$ different currencies,  let $\psi_{ij}^{(k)}$ denote the proportion of the whole capital the investor holds at the  date $k$ to buy the foreign 
	currency $j$ via the home currency $i$. Then, the  portfolio matrix $\psi^{(k)}$ for the $m$ different currencies at the  date $k$ can be formulated as 
	\begin{equation}\label{P1}
	\psi^{(k)}= \left(\begin{array}{ccccc}
	\psi_{11}^{(k)} & \psi_{12}^{(k)}&\cdots&\psi_{1m}^{(k)}\\
	\psi_{21}^{(k)} & \psi_{22}^{(k)}&\cdots&\psi_{2m}^{(k)}\\
	\vdots & \vdots&\ddots&\vdots\\
	\psi_{m1}^{(k)} & \psi_{m2}^{(k)}& \cdots&\psi_{mm}^{(k)}\\
	\end{array}
	\right).
	\end{equation}
Note that, for $i,j\in\Lambda_0$,  $\psi_{ij}^{(k)}\ge0,  \psi_{ii}^{(k)}=0, \sum_{i,j=1}^m\psi_{ij}^{(k)}=1$. In what follows,  let $P^{(k)}$ be the collection of all portfolio matrices 
for the $m$ different currencies   at the date $k$, i.e., 
	\begin{equation*} 
	P^{(k)}:=\Big\{\psi^{(k)}\in \mathbb{R}^m\otimes \mathbb{R}^m: \sum_{i,j\in\Lambda_0}\psi_{ij}^{(k)}=1,\psi_{ij}^{(k)} > 0, \psi_{ii}^{(k)}=0, i,j\in\Lambda_0\Big\},
	\end{equation*}
where $\mathbb{R}^m\otimes \mathbb{R}^m$	stands for the totality of $m\times m$-matrices. 

	\item For  $A=(a_{ij})_{m\times m}\, , B=(b_{ij})_{m\times m}\in \mathbb{R}^m\otimes \mathbb{R}^m$,  we define a product by the following  
	
	\begin{equation}\label{lv}
	A\boxtimes B:=(c_{ij})_{m\times m}\in \mathbb{R}^m\otimes \mathbb{R}^m, 
	\end{equation}
	in which $c_{ij}:=a_{ij}b_{ij}.$

	\item The return matrix $\psi^{(k)}\boxtimes R^{(k)}$ of the portfolio matrix $\psi^{(k)}$  is given by 
	\begin{equation*}
	\psi^{(k)}\boxtimes R^{(k)}=\left(\begin{array}{ccccc}
	0 & \psi_{12}^{(k)}\bar r_{12}^{(k)}&\cdots&\psi_{1m}^{(k)} \bar r_{1m}^{(k)}\\
	\psi_{21}^{(k)}\underline r_{21}^{(k)} & 0&\cdots&\psi_{2m}^{(k)}\bar r_{2m}^{(k)}\\
	\vdots & \vdots&\ddots&\vdots\\
	\psi_{m1}^{(k)}\underline r_{m1}^{(k)} & \psi_{m2}^{(k)}\underline r_{m2}^{(k)}&\cdots&0\\
	\end{array}
	\right)
	\end{equation*}
	according to \eqref{lv}.

	\item Let $F_k$ be the capital the investor holds at the end of the date $k$ and $T_k$ the transaction costs charged at the day $k$. We assume that the capital at the end of the  date $k$ is same as the begginning of the day $k+1$. Let  $F'_k$ be  the total funds at the date $k$ whenever the transaction costs are taken into consideration.  Obviously, 
	\begin{equation}\label{Bulage}
	F'_k=F_{k-1}-T_k,  \quad k\geqslant 1.\end{equation}
\item Set
\begin{equation*}
\psi^{(k)}\diamond R^{(k)}:=\left(\begin{array}{ccccc}
	1 & 1&\cdots&1\\

	\end{array}
	\right)_{(1\times m)}\Big(\psi^{(k)}\boxtimes R^{(k)}\Big)\left(\begin{array}{c}
	
	1\\
	1\\
	\vdots\\
	1\\
	\end{array}
	\right)_{(m\times 1)}.
\end{equation*}
Thus, it is obvious to see that the scalar value 
\begin{equation}\label{Qingmai}
\psi^{(k)}\diamond R^{(k)}=\sum_{i,j=1}^{m}\Big( \psi_{ij}^{(k)}\Big(\bar r^{(k)}_{ij}+\underline r^{(k)}_{ij}\Big)\Big)
\end{equation}
and, by taking the transaction costs into account, we have 
\begin{equation}\label{Fishman}
F_k=F'_k( \psi^{(k)}\diamond R^{(k)}).
\end{equation}
\end{itemize}

\subsection{The on-line portfolio strategy}
 As the high volatility of the currency exchange rates,  the investor prefers to  adjust  the investment strategy  frequently. 
 It is worthwhile pointing out that one of the  effective strategy to optimize the portfolio matrix is the on-line portfolio selection 
 by taking advantage of the historical data. More precisely,  
\begin{itemize}
\item The on-line portfolio strategy is given by 
\begin{equation*}
\psi^{(k+1)}=f \Big(\{\psi^{(1)},...,\psi^{(k)}\},\{R^{(1)},..., R^{(k)}\}\Big),
\end{equation*}
where $f:( \mathbb{R}^m\otimes \mathbb{R} ^m)^{2k}\rightarrow P^{(k+1)}$ with 
$$( \mathbb{R}^m\otimes \mathbb{R} ^m)^{2k}:=\underbrace{ \mathbb{R}^m\otimes \mathbb{R} ^m \cdots\times  \mathbb{R}^m\otimes \mathbb{R} ^m}_{2k}.$$
It is clear to see that the  on-line portfolio strategy $\psi^{(k+1)}$ depends on the portfolio matrices and the return matrices before the $(k+1)$-th trading date. 

\item Without taking  the transaction costs into consideration, the total capital   during the course of the investment period  (i.e., $N$ days) will increase  to 
\begin{equation}\label{Yulin}
I_N\Big(\{\psi^{(\cdot)}\},\{R^{(\cdot)}\}\Big):=\prod_{k=1}^N\bigg\{ \psi^{(k)}\diamond R^{(k)}\bigg\} ,
\end{equation}
where $\psi^{(k)}\diamond R^{(k)}$ is given in \eqref{Qingmai}. Herein, $I_N\Big(\{\psi^{(\cdot)}\},\{R^{(\cdot)}\}\Big)$  is referred to as the final return  
without the transaction costs. With \eqref{Yulin} in hand,  the exponential growth rate of the funds without  the transaction costs is represented by 
\begin{equation}\label{b-p}
LI_N(\{\psi^{(\cdot)}\},\{R^{(\cdot)}\}) := \frac{1}{N}\sum_{k=1}^{N}\log \bigg( \psi^{(k)}\diamond R^{(k)}\bigg).
\end{equation}

\item Let $(c_k)_{k\in\Lambda_1}$ with $c_1=0$ be the ratio of the transaction cost at the  date $k$ (i.e., $T_k$) to the fund holding at the  beginning of the date $k$ (i.e., $F_{k-1}$), in other words, 
\begin{equation}\label{Hefei}
c_1=0; \qquad
c_k=\frac{T_k}{F_{k-1}}=\frac{F_{k-1}-F'_k}{F_{k-1}},~~~~~k\ge2,
\end{equation}
where the second identity for $c_k$ is due to \eqref{Bulage} and $F_k>0$ for $k\ge1.$

\item In the presence of the transaction costs,  the return of the total investment during the investment period  will reach to 
\begin{equation}\label{Huashan}
F_N(\{\psi^{(\cdot)}\},\{R^{(\cdot)}\}) = \prod_{k=1}^N\Big(\psi^{(k)}\diamond R^{(k)}\Big)(1-c_k),
\end{equation}
where $c_k$ is defined as in \eqref{Hefei}.  In \eqref{Huashan},  $F_N(\{\psi^{(\cdot)}\},\{R^{(\cdot)}\})$ is named as the return with the transaction costs. With the \eqref{Huashan},  we formulate the 
exponential growth rate of the funds with transaction costs as 
\begin{equation}\label{P-B}
R_N(\{\psi^{(\cdot)}\},\{R^{(\cdot)}\}) := \frac{1}{N}\sum_{k=1}^{N}\log \bigg( \psi^{(k)}\diamond R^{(k)}\bigg)+\frac{1}{N}\sum_{k=1}^{N}\log (1-c_k).
\end{equation}
\end{itemize}

\subsection{Transaction costs}

According to the self-financing strategy,  the investor will reinvest all of the funds held at the end of the trading date $k-1, \mbox{for }  k\geq 2$ to the beginning of the next trading date $k$. In our 
present case, we shall consider the transaction costs at the end of the trading date. Hence, the total funds will be reduced.  For more details, the reader is referred to \eqref{Bulage}.   In the light of \eqref{Fishman}, the investment on the currency pair $(i,j) $  admits the form 
\begin{equation}\label{Swansea}
X^{k}_{ij}=F'_k \psi^{(k)}_{ij}\Big(\bar r^{(k)}_{ij}+\underline r^{(k)}_{ij}\Big),\end{equation}
where $F'_k$ is introduced in \eqref{Bulage}. At the trading date $k+1$,   the investor will hold a new portfolio matrix $\psi^{(k+1)}$. Meanwhile, at the date $k+1,$ the transaction cost  $T_{k+1}$ will 
be charged by the bank involved so that the funds at the date $k+1$ is $F'_{k+1}=F_k-T_{k+1}$. Thus,  the  investment for the currency pair $(i,j)$ at the date $k+1$ is
\begin{equation}\label{Oxford}
X'^{(k+1)}_{ij}=F'_{k+1}\psi^{(k+1)}_{ij}.\end{equation}
Henceforth,  from \eqref{Swansea} and \eqref{Oxford}, it follows that  the transaction cost between the currency pair   $(i,j)$  is
\begin{equation}\label{Cambridge}|X'^{(k+1)}_{ij} -X^{k}_{ij}|=\Big|F'_{k+1}\psi^{(k+1)}_{ij}-F'_k \psi^{(k)}_{ij}\Big(\bar r^{(k)}_{ij}+\underline r^{(k)}_{ij}\Big)\Big|.\end{equation}
We further suppose (for simplicity) that the transaction costs are the same for selling and buying foreign currencies. As a consequence, \eqref{Bulage} and  \eqref{Cambridge} imply that  the total 
transaction costs at the start of the date $(k+1)$ is  
\begin{equation}\label{panpan}
\begin{split}
Q_{k+1}=\sum_{i,j=1}^{m} |X'^{(k+1)}_{ij} -X^{k}_{ij}|
&=\sum_{i,j=1}^{m}\Big |F'_{k+1}\psi^{(k+1)}_{ij}-F'_k \psi^{(k)}_{ij}\Big(\bar r^{(k)}_{ij}+\underline r^{(k)}_{ij}\Big)\Big|\\
&=\sum_{i,j=1}^{m}\Big |(F_k-T_{k+1})\psi^{(k+1)}_{ij}-F'_k \psi^{(k)}_{ij}\Big(\bar r^{(k)}_{ij}+\underline r^{(k)}_{ij}\Big)\Big|.
\end{split}
\end{equation}
In the case that the transaction costs $T_{k+1}$ is linearly dependent  on the total investment $Q_{k+1}$ in the trading date $k+1$, there exists a positive constant $c$ such that $T_{k+1}=cQ_{k+1}$.  Accordingly, \eqref{panpan} yields that 
\begin{equation}\label{decrement}
T_{k+1} =cQ_{k+1}=
\sum_{i,j=1}^{m} \Big|F_{k}\psi^{(k+1)}_{ij}-F'_k \psi^{(k)}_{ij}\Big(\bar r^{(k)}_{ij}+\underline r^{(k)}_{ij}\Big)-T_{k+1}\psi^{(k+1)}_{ij}\Big|.
\end{equation}
For  $\psi^{(k+1)}\in P^{(k+1)}$ and $T_{k+1}\geqslant 0$, we set $$\Delta _{k+1}:=\sum_{i,j=1}^{m} \Big|F_k\psi^{(k+1)}_{ij}-F'_k\psi^{(k)}_{ij}\Big(\bar r^{(k)}_{ij}+\underline r^{(k)}_{ij}\Big)\Big|.$$ 
It is easy to see that
\begin{equation*} 
c\Big(\Delta _{k+1}-\sum_{i,j=1}^{m} T_{k+1}\psi^{(k+1)}_{ij}\Big)\le T_{k+1} \le c\Big(\Delta _{k+1}+\sum_{i,j=1}^{m} T_{k+1}\psi^{(k+1)}_{ij}\Big).
\end{equation*}
This, together with $\psi^{(k+1)}\in P^{(k+1)}$ and $\sum_{i,j=1}^{m} \psi^{(k+1)}_{ij}=1$,  leads to  
\begin{equation}\label{inq1}
c\Big(\Delta _{k+1}-T_{k+1}\Big)\le T_{k+1} \le c\Big(\Delta _{k+1}+\ T_{k+1}\Big).
\end{equation}
If  $c<1$,  then, we derive from   \eqref{inq1}  that 

\begin{equation}\label{ineq1}
\frac{c}{1+c}\Delta_{k+1}  \leqslant T_{k+1} \leqslant \frac{c}{1-c}\Delta_{k+1}.
\end{equation}
Let 
\begin{equation*}
	{\psi'}^{(k)}:= \left(\begin{array}{ccccc}
	{\psi'}_{11}^{(k)} & {\psi'}_{12}^{(k)}&\cdots&{\psi'}_{1m}^{(k)}\\
	{\psi'}_{21}^{(k)} & {\psi'}_{22}^{(k)}&\cdots&{\psi'}_{2m}^{(k)}\\
	\vdots & \vdots& \ddots&\vdots\\
	{\psi'}_{m1}^{(k)} & {\psi'}_{m2}^{(k)}&\cdots&{\psi'}_{mm}^{(k)}\\
	\end{array}
	\right),
	\end{equation*}
with  	
\begin{equation}\label{p1}
{\psi'}^{(k)}_{ij}:=\frac{F'_k}{F_k}\psi^{(k)}_{ij}\Big(\bar r^{(k)}_{ij}+\underline r^{(k)}_{ij}\Big).
\end{equation}	
We remark that ${\psi'}^{(k)}$ is the proportion matrix for the $m$ different currencies reached automatically  at the end of the date $k$. 

In what follows, we assume that $F_k>0$  and $\psi^{(k)}\diamond R^{(k)}>0.$ 
We define the {\it distance} between $\psi^{(k+1)}$ and $\psi'^{(k)}$  by
\begin{equation}\label{distance}
d(\psi^{(k+1)},\psi'^k):=\sum_{i,j=1}^{m}|\psi^{(k+1)}_{ij}-\psi'^{(k)}_{ij}|.
\end{equation}
By virtue of the notion of $\Delta _{k+1}$, one has 
\begin{equation*}
\begin{split}
\Delta _{k+1}=F_k\sum_{i,j=1}^{m} \Big|\psi^{(k+1)}_{ij}-\frac{F'_k}{F_k}\psi^{(k)}_{ij}\Big(\bar r^{(k)}_{ij}+\underline r^{(k)}_{ij}\Big)\Big)\Big|.
\end{split}
\end{equation*}
This, combining with \eqref{p1}, gives that 
\begin{equation*}
\Delta _{k+1}=F_k\sum_{i,j=1}^{m} \Big|\psi^{(k+1)}_{ij}-{\psi'}^{(k)}_{ij}\Big|=F_kd(\psi^{(k+1)},\psi'^{(k)}).
\end{equation*}
Substituting this into \eqref{ineq1}, we end up with the following 
\begin{equation}\label{ineq2}
\frac{c}{1+c} F_kd(\psi^{(k+1)},\psi'^{(k)}) \leqslant T_{k+1} \leqslant \frac{c}{1-c}F_kd(\psi^{(k+1)},\psi'^{(k)}).
\end{equation}

From (\ref{ineq2}),  we observe that the transaction cost at the date $k+1$ depends on the distance between $\psi^{(k+1)}$ and $\psi'^{(k)}$, and that  the  bigger distance 
between $\psi^{(k+1)}$ and $\psi'^{(k)}$ means the more transaction costs which further implies  less profit in the portfolio. Hence, to improve the profit for  the portfolios, it is 
essential for the investors  to shorten the distance between  $\psi^{(k+1)}$ and $\psi'^{(k)}$ . 

\section{Update rules for on-line portfolio selections}
As the high volatility of the currency exchange rates, the investors, in general,  try to buy and sell the currencies again and again to get more profits. Unfortunately, the more transactions 
means the more transaction costs. So, it is indispensable to optimise the portfolios in order  to evade  the unnecessary transaction costs, which is our goal in this section. 

Let $Z(\psi^{(k+1)})$ be the return of the portfolio $\psi^{(k+1)}$ at the date $k+1$ and assume that $Z(\cdot)$ admits the following form 
\begin{equation} \label{profit}
Z(\psi^{(k+1)})=\gamma Z_F(\psi^{(k+1)},R'^{(k+1)})-Z_T(\psi^{(k+1)}) , 
\end{equation}
where  
\begin{itemize}
\item \begin{equation}
R'^{(k+1)}:= \left(\begin{array}{ccccc}
	r_{11}^{'(k+1)} & \bar r_{12}^{'(k+1)}& \cdots&\bar r_{1m}^{'(k+1)}\\
	\underline r_{21}^{'(k+1)} & r_{22}^{'(k+1)}& \cdots&\bar r_{2m}^{'(k+1)}\\
	\vdots & \vdots& \ddots&\vdots\\
	\underline r_{m1}^{'(k+1)} &\underline  r_{m2}^{'(k+1)}&\cdots&r_{mm}^{'(k+1)}\\
	\end{array}
	\right)
	\end{equation}is the prediction of return matrix $R^{(k+1)}$ at the date $k+1;$

\item $ Z_F(\psi^{(k+1)},R'^{(k+1)})$ is a function of  investment increments with respect to the portfolio matrix  $\psi^{(k+1)} $ and the prediction $R'^{(k+1)}$ of $R^{(k+1)}$; 

\item $ Z_T(\psi^{(k+1)} )$ is the transaction cost which the  investor  pays  for the portfolio matrix $\psi^{(k+1)} $;

\item $\gamma>0$ is a parameter adopted to balance maximizing the investment increase and reducing the transactions. 
\end{itemize}
                                                                                                                                                                                                                                                                                                                                                                                                                                                                                                                                                                                                                                                                                                                                                                                                                                                                                                                                                                                                                                                                                                                                                                                                                                                                                                                                                                                                                                                                                                                                                                                                                                                                                                                                                                                                                                                                                                                                    In the present paper, the first alternative for the function $ Z_F$, denoted by $Z_F^{(1)}$,  admits the form 
\begin{equation}\label{z1}
Z_F^{(1)}(\psi^{(k+1)},R'^{(k+1)})=\psi^{(k+1)}\diamond R'^{(k+1)}.
\end{equation}
And the second choice for $ Z_F$, written by $Z_F^{(2)}$,  possesses the following representation 
\begin{equation}\label{z2}
\begin{split}
Z_F^{(2)}\Big(\psi^{(k+1)},R'^{(k+1)}\Big)=\log\Big(\psi'^{(k)}\diamond R'^{(k+1)}\Big)+\frac{R'^{(k+1)}\Big(\psi^{(k+1)}-\psi'^{(k)}\Big)}{\psi^{(k+1)}\diamond R'^{(k+1)}}.
\end{split}
\end{equation}
By the Taylor expansion formula  for the multivariate functions, we deduce that 
\begin{equation*}
\begin{split}
\log(\psi^{(k+1)}\diamond R'^{(k+1)})&\approx\log \bigg(\sum_{i,j=1}^{m}\Big( {\psi'}_{ij}^{(k)}(\bar r_{ij}^{'(k+1)}+\underline r_{ij}^{'(k+1)})\Big)\bigg)\\
&\quad+\frac{\sum_{i,j=1}^{m}(\bar r_{ij}^{'(k+1)}+\underline r_{ij}^{'(k+1)})
\Big(\psi_{ij}^{(k+1)}-\psi_{ij}'^{(k)}\Big)}{\sum_{i,j=1}^{m}\psi_{ij}^{(k+1)}(\bar r_{ij}^{'(k+1)}+\underline r_{ij}^{'(k+1)})}.
\end{split}
\end{equation*}
Therefore, $Z_F^{(2)}\Big(\psi^{(k+1)},R'^{(k+1)}\Big)$ is the first order approximation of $\log(\psi^{(k+1)}\diamond R'^{(k+1)})$.

 With regard to the term  $Z_T(\psi^{(k+1)})$, we define 
 \begin{equation}\label{Xibei}
 Z_	T(\psi^{(k+1)}):=d_{\mbox{re}}\Big(\psi^{(k+1)},\psi'^{(k)}\Big). 
 \end{equation}
 We refer the reader to \cite{ALZ2001, helmbold, kivinen} for more details. According to the definition of relative entropy for discrete random variables, we have the following 
\begin{align}\label{zd}
\begin{split}
d_{\mbox{re}}\Big(\psi^{(k+1)},\psi'^{(k)}\Big)&=\sum_{i,j=1}^{m}\bigg(\psi^{(k+1)}_{ij}\log\frac{\psi^{(k+1)}_{ij}}{\psi'^{(k)}_{ij}}\bigg)\\
&=\sum_{i,j=1}^{m}\bigg(\psi^{(k+1)}_{ij}\log\frac{\psi^{(k+1)}_{ij}\sum_{i,j=1}^{m}\Big( \psi_{ij}^{(k)}(\bar r^{(k)}_{ij}+\underline r^{(k)}_{ij})\Big)}{\psi^{(k)}_{ij}(\bar r^{(k)}_{ij}+\underline r^{(k)}_{ij})}\bigg).
\end{split}
\end{align}
By L'Hospital's rule, one has $\lim_{x\downarrow0}x\log x=0$. So, without loss of generality,  in \eqref{zd}, we can assume $\psi^{(k+1)}_{ij}\log\psi^{(k+1)}_{ij}=0$ whenever $\psi^{(k+1)}_{ij}=0.$
Note that for any constant $a>1,f(x)=x(\log x+\log a)$ is a convex function for $x>0$
due to the fact that $f''(x)>0$. Therefore,
$d_{\mbox{re}}\Big(\psi^{(k+1)},\psi'^{(k)}\Big)$ is  a positive continuous
convex function of $\psi^{(k+1)}_{ij}, 1\le i,j\le m$. If $ \psi^{(k+1)}=\psi'^{(k)}$, 
i.e., $\psi^{(k+1)}_{ij}=\psi'^{(k)}_{ij}$, then
$d_{\mbox{re}}\Big(\psi^{(k+1)},\psi'^{(k)}\Big)=0$. The minimum of
$d_{\mbox{re}}\Big(\psi^{(k+1)},\psi'^{(k)}\Big)$ can be achieved at $
\psi^{(k+1)}=\psi'^{(k)}$. Moreover, let $i_0,j_0\ge1$ be such that
$\psi'^{(k)}_{i_0j_0}=\min\{\psi'^{(k)}_{ij}\}$. Then the maximum of
$d_{\mbox{re}}\Big(\psi^{(k+1)},\psi'^{(k)}\Big)$ can be available whenever
$\psi^{(k+1)}_{ij}=1$ for $i=i_0,j=j_0$, otherwise
$\psi^{(k+1)}_{ij}=0$.
Inserting \eqref{z1}, \eqref{z2} and \eqref{Xibei} back into \eqref{profit}, respectively, we arrive at 
\begin{equation}\label{z111}
Z'^{(1)}(\psi^{(k+1)})=\gamma\Big( \psi^{k+1}\diamond R'^{(k+1)}\Big)-\sum_{i,j=1}^{m}\bigg(\psi^{(k+1)}_{ij}\log\frac{\psi^{(k+1)}_{ij}}{\psi'^{(k)}_{ij}}\bigg) 
\end{equation}
and 
\begin{equation}\label{z221}
\begin{split}
Z'^{(2)}(\psi^{(k+1)})&=\gamma \log\Big(\psi'^{(k)}\diamond R'^{(k+1)}\Big)+\frac{R'^{(k+1)}\Big(\psi^{(k+1)}-\psi'^{(k)}\Big)}{\psi^{(k+1)}\diamond R'^{(k+1)}}\\
&\quad-\sum_{i,j=1}^{m}\bigg(\psi^{(k+1)}_{ij}\log\frac{\psi^{(k+1)}_{ij}}{\psi'^{(k)}_{ij}}\bigg). 
\end{split}
\end{equation}
As $d_{\mbox{re}}$ is convex (as shown above) with respect to $\psi^{(k+1)}_{ij}$, $-d_{\mbox{re}}$ is concave with respect to $\psi^{(k+1)}_{ij}$. Observe that, except the entropy term $d_{\mbox{re}}$, 
the other terms are linear with respect to $\psi^{(k+1)}_{ij}$, which are obviously concave. As a result, we conclude that $Z'^{(1)}(\psi^{(k+1)})$ and $Z'^{(2)} (\psi^{(k+1)})$, defined 
in \eqref{z111} and \eqref{z221}, respectively,  are concave functions with respect to the portfolio matrix entries $\psi^{(k+1)}_{ij}$.

To maximise $Z'^{(1)}(\psi^{(k+1)})$ with the constraint $\psi^{(k+1)}\in P^{(k+1)}$ (so that $\sum_{i,j=1}^m\psi^{(k+1)}_{ij}=1$), we utilise Lagrange's method. 
To this end,  we consider the following auxiliary function
\begin{equation*}
\begin{split}
Z'^{(1)}(\psi^{(k+1)},\lambda)&=Z'^{(1)}(\psi^{(k+1)})+\lambda\Big(\sum_{i,j=1}^m\psi^{(k+1)}_{ij}-1\Big) 
\end{split}
\end{equation*}
where $\lambda\in  R$ is the Lagrange multiplier. According to \eqref{Qingmai} and \eqref{z111}, we obtain that 
\begin{equation}\label{baby}
\begin{split}
Z'^{(1)}(\psi^{(k+1)},\lambda)&=\gamma\sum_{i,j=1}^{m}\psi^{(k+1)}_{ij}\Big(\bar r_{ij}^{'(k+1)}+\underline r_{ij}^{'(k+1)}\Big)-\sum_{i,j=1}^{m}\bigg(\psi^{(k+1)}_{ij}\log\frac{\psi^{(k+1)}_{ij}}{\psi'^{(k)}_{ij}}\bigg)\\
&\quad+\lambda\Big(\sum_{i,j=1}^m\psi^{(k+1)}_{ij}-1\Big).
\end{split}
\end{equation}
Taking derivatives with respect to the variables $\psi^{(k+1)}_{ij}$ and $\lambda$ for $Z'^{(1)}(\psi^{(k+1)},\lambda)$ followed by letting $\frac{\partial }{\partial \psi^{(k+1)}_{ij}}Z'^{(1)}(\psi^{(k+1)},\lambda)=\frac{\partial }{\partial \lambda }Z'^{(1)}(\psi^{(k+1)},\lambda)=0$, we then have 
\begin{equation*}
\begin{split}
\gamma\Big(\bar r_{ij}^{'(k+1)}+\underline r_{ij}^{'(k+1)}\Big)-\log\psi^{(k+1)}_{ij} +\log\psi'^{(k)}_{ij}-1+\lambda=0.
\end{split}
\end{equation*}
This further implies that 
\begin{equation}\label{ba}
\begin{split}
\psi^{(k+1)}_{ij}&=\exp\bigg(\gamma\Big(\bar r_{ij}^{'(k+1)}+\underline r_{ij}^{'(k+1)}\Big)+\log\psi'^{(k)}_{ij}-1+\lambda\bigg)\\
&=\frac{\psi'^{(k)}_{ij}\exp\bigg(\gamma\Big(\bar r_{ij}^{'(k+1)}+\underline r_{ij}^{'(k+1)}\Big)\bigg)}{\e^{1-\lambda}}.
\end{split}
\end{equation}
In view of  $\sum_{i,j=1}^m\psi^{(k+1)}_{ij}=1$, we get from \eqref{ba} that 
\begin{equation*} 
\sum_{l,v=1}^m\frac{\psi'^{(k)}_{lv}\exp\bigg(\gamma\Big(\bar r_{lv}^{'(k+1)}+\underline r_{lv}^{'(k+1)}\Big)\bigg)}{\e^{1-\lambda}}=1.
\end{equation*}
Therefore, it follows that
\begin{equation*} 
\e^{1-\lambda}=\sum_{l,v=1}^m\psi'^{(k)}_{lv}\exp\bigg(\gamma\Big(\bar r_{lv}^{'(k+1)}+\underline r_{lv}^{'(k+1)}\Big)\bigg).
\end{equation*}
Substituting this into \eqref{ba} leads to 
\begin{equation}\label{portf}
\psi^{(k+1)}_{ij}=\frac{\psi'^{(k)}_{ij}\exp(\gamma (\bar r_{ij}^{'(k+1)}+\underline r_{ij}^{'(k+1)}))}{\sum_{v,l=1}^{m}\psi'^{(k)}_{vl}\exp(\gamma (\bar r_{vl}^{'(k+1)}+\underline r_{vl}^{'(k+1)}))}.
\end{equation}
(\ref{portf}) is the update rule of the Increment of the Investment with Transaction Cost (IITC for abbreviation) for the $(k+1)^{th}$ trading day.

Mimicking the procedure for  the derivation of \eqref{portf},  one can conclude that $Z'^{(2)}(\psi^{(k+1)})$ reaches its maximum at 
\begin{equation}\label{portf1}
\psi^{(k+1)}_{ij}=\frac{\psi'^{(k)}_{ij}\exp\Big(\frac{\gamma (\bar r_{ij}^{'(k+1)}+\underline r_{ij}^{'(k+1)})}{\psi'^{(k)} \diamond R'^{(k+1)}}\Big)}{\sum_{v,l=1}^{m}\psi'^{(k)}_{vl}
\exp\Big(\frac{\gamma (\bar r_{vl}^{'(k+1)}+\underline r_{vl}^{'(k+1)})}{\psi'^{(k)} \diamond R'^{(k+1)}}\Big)}.
\end{equation}
In our case,  \eqref{portf1}  is the update rule of the Exponential  Increment of the Investment with Transaction Cost (EIITC for short) for the  trading day $k+1$.

It is clear to see  that, in (\ref{portf}) and (\ref{portf1}), there are two parameters  $\gamma$ and $R'^{(k+1)}$ to be selected for the portfolio matrix $\psi^{(k+1)}$. Concerning the variable  
$\gamma$, different values yield different strategies. We would like to explicate a bit more details about the implications of the parameter $\gamma$. More precisely, (i)  $\gamma$ stands 
for the passive strategy; (ii) the smaller $\gamma$ signifies a weaker prediction of the portfolio matrix $\psi^{(k+1)}$ so that the investors prefer to hold the present portfolio matrix (i.e.,$\psi^{(k)}$) 
to avoid the decrements; (iii) the bigger $\gamma$ represents a stronger prediction of  the portfolio matrix $\psi^{(k+1)}$  so as to the investors prefer  to change the present portfolio  matrix (i.e., 
$\psi^{(k)}$) to earn more profits. With regard to $R'^{(k+1)}$, it is the prediction of the return at date $k+1$. Clearly, a high quality prediction is a power tool for investors to make a profitable 
decision  for their investments.

Before ending up this section, let us give some remarks. 
\begin{remark}
The quantity $Z_F^{(1)}(\psi^{(k+1)},R'^{(k+1)})$ indicates the increment of the fund at the date $k+1$. On the other hand, one can apply the distance of portfolio matrices  $\psi'^{(k)}$ and
$\psi^{(k+1)}$, defined in \eqref{distance}, to measure the decrement function.
\end{remark}

\begin{remark}
Apparently, in the realistic financial market, the investors prefer to give up the unprofitable transactions and to add profitable ones instead, in order to avoid the decrements. Consequently, 
the update rules \eqref{portf}  and \eqref{portf1}  can be attainable.
\end{remark}

\section{Prediction of the returns}
 
Prediction of the returns for a portfolio plays a vital role in  optimising  the portfolios in the foreign exchange markets. Motivated by \cite{ALZ2001}, in this paper,  we shall establish 
an algorithm to keep and/or to inject profitable pairs of currencies and to remove unprofitable pairs, which will be called the {\it cross rate} algorithm. 

In terms of the  mechanism of the foreign exchange market, some entries in the return matrix are vanished, see Equations \eqref{ce1} - \eqref{r2}. More precisely,  for any 
$(i,j)\in\Lambda_0\times\Lambda_0$,  $ \bar r_{ij}^{(k)}=0$ whenever $\underline{r}_{ij}^{(k)}\neq 0$, while $ \bar r_{ij}^{(k)}\neq0$ for $\underline{r}_{ij}^{(k)}=0$, and moreover 
$r_{ii}^{(k)}=0$. Therefore, there are $\frac{m(m-1)}{2}$ non-zero components  in the return matrix. As it is known, in the foreign currency market, the best profitable pair of currencies 
means the value of the return in the return matrix reaches to the maximum. In the sequel, let $R^{(k)}$ be the return matrix of the $m$ different currencies.  Let 
$\alpha:=\max\{\bar r_{ij}^{(k)}, \underline r_{ij}^{(k)}, \, i,j\in\Lambda_0\}$. Now, we define the order  $ O(R^{(k)})$ of $R^{(k)}$ as follows  

\begin{equation}\label{order}
O\Big(R^{(k)}\Big) =
\begin{cases}
1,  ~~~~~~~~~\mbox{ there is only one pair  } (i,j)\in \Lambda_0\times\Lambda_0 \mbox{ such that } \bar r_{ij}^{(k)}=\alpha\\
2,  ~~~~~~~~~\mbox{ there is only one pair  } (i,j)\in \Lambda_0\times\Lambda_0 \mbox{ such that } \underline r_{ij}^{(k)}=\alpha\\
0,  ~~~~~~~~~\mbox{ the others}.
\end{cases}\end{equation}
In fact, our order is based on the trading action which are either buying or selling foreign currencies. More precisely, if the order is $1$, the investor is going to sell the foreign currencies 
in terms of  the only one maximum component of return matrix appearing in the upper triangular part of $R^{(k)}$ (namely, $\max\{\bar r^{(k)}_{ij}\}>\max\{\underline r^{(k)}_{ij}\}$); while if the order 
is $2$, the investor is going to buy the foreign currencies in terms of the only one maximum component of return matrix appearing in the lower triangular part of  $R^{(k)}$  
(i.e., $\max\{\bar r^{(k)}_{ij}\}<\max\{\underline r^{(k)}_{ij}\}$). The order $0$ means there is no action for the investor at all. 

We call the return sequence $\{R^{(k)}\}_{\{1\le k\le N\}}$  is {\it strictly unequal} if 
\begin{equation*}
O(R^{(k)})\neq 0,~~~~~k\in \Lambda_1.
\end{equation*}
Define
\begin{equation*}
\mbox{Rev} (R^{(k)})=(R^{(k)})^T,
\end{equation*}
where $(R^{(k)})^T$ denotes the transpose of $R^{(k)}$. 
If  $ O(R^{(k)})\neq O(R^{(k-1)}) $  for some $k\in \Lambda_1$, then $k$ is called a {\it cross position}.  For $E,F\in \Lambda_0$ with $E<F$, set 
\begin{equation}\label{Marriott}
 \acute{D}(E,F):= \{R^{(k)},  k\in \{E, E+1, \cdots, F\}\}.
\end{equation}
Observe  that $\acute{D}(E,F)$ is the collection of all return matrices of the currencies involved  from the day $E$ to the day $F.$ Let 
 \begin{equation*}
C_{E,F}:= \sharp \{k:O(R^{(k)})\neq O(R^{(k-1)}), ~~k\in \{E, E+1, \cdots, F\}\},
\end{equation*}
 which counts the number of the cross positions from the day $E$ to the day $F$, where $\sharp \{...\}$ stands for the cardinal number of the set $\{...\}$. Moreover, $C_{E,F}$  defined above is 
 named as the {\it cross number} associated with the segment  $\acute{D}(E,F).$ Let \begin{equation}\label{Marina}
W_{E,F}:= \frac{C_{E,F}}{E-F+1},
\end{equation}
which is the proportion possessed  by the cross positions during the course of the day $E$ to the day $F$. We call $W_{E,F}$ the {\it cross rate} of the segment $\acute{D}(E,F)$. 
Let $L \ge1$ be a fixed integer. We then divide the investment period into different segments with the same length $L$ in the following manner. Taking $E=(n-1)L+1$ and $F=nL$ 
in \eqref{Marriott}, one has  
 \begin{equation*} 
 \acute{D}((n-1)L+1,nL)= \{R^{(k)},  k\in \{(n-1)L+1, (n-1)L, \cdots, nL\}\}.
\end{equation*}
In what follows, we shall write $\acute{D}_n((L)$  and $W_n(L)$ in lieu  of $ \acute{D}((n-1)L+1,nL)$  and $W_{(n-1)L+1,nL}$, respectively,  for brevity of notation.  That is, 
 \begin{equation*}
\acute{D}_n((L)= \acute{D}((n-1)L+1,nL)~~~~~\mbox{ and }~~~~~~~W_n(L)= W_{(n-1)L+1,nL}.
\end{equation*}
We see that $W_n(L)\in\{0,\frac{1}{L},\frac{2}{L},...,\frac{L-1}{L},1\}\subset[0,1]$,  so one can take $\Omega=[0,1]$ as a probability space endowed with the Borel $\sigma$-algebra 
$\mathcal{B}([0,1])$ and uniform probability measure $\mathbb{P}$. Then $W_n(L)$ is nothing but a (discrete) random variable on the Borel probability space $([0,1],\mathcal{B}([0,1]),\mathbb{P})$.    

\subsection{The cross rate approach}

Let $A:=[0,\frac{1}{2})$ and  $B:=[\frac{1}{2},1]$.  We define  
\begin{equation*}
P_{AB}(n):= \mathbb{P}(W_n(L)\in A,W_{n+1}(L)\in B).
\end{equation*}
 One can define $P_{AA}(n), P_{BA}(n) $ and $P_{BB}(n)$ similarly.
 
Here and in the sequel, we assume that $W_n(L)$ is stationary (i.e., $W_n(L)$ does not change with the shift of the parameter $n$)  so that $P_{AB}(n)$ (resp. $P_{AA}(n) $, $P_{BA}(n)$, and $P_{BB}(n)$)  is independent of $n$. Therefore, we can write $P_{AB}$ (resp.  $P_{AA} $, $P_{BA}$, and $P_{BB}$) instead of $P_{AB}(n)$ (resp.  $P_{AA}(n , P_{BA}(n)$, and $P_{BB}(n)$).  Observe that
  \begin{equation*}
\begin{split}  
1&=\mathbb{P}(W_{n+1}(L)\in [0,1], W_n(L)\in [0,1]) \\
&=\mathbb{P}(W_{n+1}(L)\in A, W_n(L)\in [0,1])+\mathbb{P}(W_{n+1}(L)\in B, W_n(L)\in [0,1])\\
&=\mathbb{P}(W_{n+1}(L)\in A, W_n(L)\in A)+\mathbb{P}(W_{n+1}(L)\in A, W_n(L)\in B)\\
&\quad+\mathbb{P}(W_{n+1}(L)\in B, W_n(L)\in A)+\mathbb{P}(W_{n+1}(L)\in B, W_n(L)\in B).
 \end{split}
 \end{equation*}
 Thus, one clearly has 
\begin{equation*}
P_{AA}+P_{AB}+P_{BA}+P_{BB}=1.
\end{equation*}
 
Next, we follow the three steps below to predict  the return matrix  $ R^{(k+1)} \in$ \'{D}$_{n+1}(L)$.
\begin{itemize}
\item[{\bf Step 1: }]{\bf  The prediction of  $W_{n+1}(L)$}, denoted by $W'_{n+1}(L)$.  So far, there are several {\it methods  to predict the cross rate} (MPCR for abbreviation)  $W_{n+1}(L)$  
via $W_i(L),i=1,...,n$ for  $n\in \mathbb{N}$.  For more details, the reader is referred to \cite{ALZ2001,PanpanRenWu}. In this paper, we are interested in  the following two strategies: 
 for $P_{AA}+P_{BB}\geqslant \frac{1}{2}$, 
\begin{equation*}
\textbf{MPCR1:}\,\, ~~~~~~~~W'_{n+1}(L)=W(L), ~~~~~~~~n\in \mathbb{N}
\end{equation*}
and,  for $P_{AB}+P_{BA}\geqslant \frac{1}{2}$,
\begin{equation*}
\textbf{MPCR2:} \,\,~~~~~~~
W'_{n+1}(L)=\begin{cases}
c_A, &\text{if }W_n(L)\in B\\
c_B, &\text{if }W_n(L)\in A,
\end{cases}
\end{equation*}
where $c_A\in A$ and $c_B\in B$ are some  constants.
\item[{\bf Step 2:}] {\bf  The prediction of  the order  for} $R^{(k+1)}\in $ \'{D}$_{n+1}(L)$, denoted by $O'(R^{(k+1)})$. 		
There are two approaches which can be used to predict the order of $R^{(k+1)}$ and we listed them below
\begin{equation*}
\begin{split}
\textbf{MPO1:}  ~~~O'(R^{(k+1)})=\begin{cases}
O( \mbox{Rev}(R^{(k)})), &~~~~~~~\text{if }W'_{n+1}(L)\in [\frac{1}{2},1]\\
O(R^{(k)}), &~~~~~~~\text{if }W'_{n+1}(L)\in [0,\frac{1}{2}]
\end{cases}
\end{split}
\end{equation*}
and
\begin{equation*}      
\textbf{MPO2:} \,\,\,~~O'(R^{(k+1)})=\begin{cases}
O(R^{(k-1)}), &\text{if }W'_{n+1}(L)\in [\frac{1}{2},1]\\
O(R^{(k)}), &\text{if }W'_{n+1}(L)\in [0,\frac{1}{2}].
\end{cases}
\end{equation*}
 
\item[{\bf Step 3:}] {\bf The prediction of $R^{(k+1)}$}, denoted by $R'^{(k+1)}$. Concerning  MPO1, 
\begin{equation*}
R'^{(k+1)}=\begin{cases}
\mbox{ Rev}(R^{(k)}), &\text{if }W'_{n+1}(L)\in [\frac{1}{2},1]\\
R^{(k)}, &\text{if }W'_{n+1}(L)\in [0,\frac{1}{2}]
\end{cases}
\end{equation*}
and,  for MPO2 , 
\begin{equation*}
R'^{(k+1)}=\begin{cases}
R^{(k-1)}, &\text{if }W'_{n+1}(L)\in [\frac{1}{2},1]\\
R^{(k)} &\text{if }W'_{n+1}(L)\in [0,\frac{1}{2}]\, .
\end{cases}
\end{equation*}
 
\end{itemize}

The three procedures above applied to obtain the prediction $R'^{(k+1)}$ of $R^{(k+1)}$  via MPCR and MPO  is called a {\it cross rate} method, which is denoted by CR(MPCR, MPO, $R'^{(k+1)}$).

\subsection{The adjusted cross rate method}
In the previous subsection, we consider only the case that the return sequence $\{R^{(k)}\}_{1\le k\le N}$  is strictly unequal. In this subsection, we move forward to investigate the setting which 
allows the return sequence need not to be strictly unequal. To cope with this setup, we need to  adjust the cross rate method introduced previously.  For $E,F\in \Lambda_0$ with $E<F$, set 
\begin{equation*}
C'_{E,F}:=\sharp \{k:O(R^{l(k)})\neq O(R^{(k)})\text{ and }  O(R^{(k)})\neq 0, E\leqslant k\leqslant F\},
\end{equation*}
where
\begin{equation*}
l(k):=\max\{l:l<k,O(R^l)\neq 0\}.
\end{equation*}
According to the definition of $C'_{E,F}$, $R^{l(k)}$ is the return matrix which is nearest to $R^{(k)}$, where $O(R^{(k)})\neq 0$,  and whose order is different from  that of $R^{(k)}$ . 

Define  the  cross rate $W'_{E,F} $ of $ \acute{D}(E,F)$, introduced in \eqref{Marriott}, by 
\begin{equation}\label{Debenham}
W'_{E,F} =\frac{C'_{E,F}}{n_{E,F}},
\end{equation}
where
\begin{equation*}
n_{E,F}:=\sharp\{k:O(R^{(k)})\neq \Delta, E\leqslant k\leqslant F\}.
\end{equation*}
In \eqref{Debenham}, choosing $E=(n-1)L+1$ and $F=nL$, we have 
\begin{equation*} 
W'_{(n-1)L+1,nL} =\frac{C'_{(n-1)L+1,nL}}{n_{(n-1)L+1,nL}}. 
\end{equation*}
In the sequel, for notation simplicity, we shall write $ W'_{L} $ instead of $W'_{(n-1)L+1,nL} $. 

Following the procedure of  CR(MPCR, MPO, $R'^{(k+1)}$) for the strictly unequal framework, we adopt the following steps to 
predict the return matrix $R^{(k+1)}$. 
  
\begin{itemize}
\item[{\bf Step 1:}] {\bf The prediction of $W'_{n+1}(L)$}, denoted by $W''_{n+1}(L)$.  There are two methods: 
\begin{equation}
\textbf{MPCR1$'$:} \,\, ~~~~~~~~~W''_{n+1}(L)=W(L), ~~~~~n\in \mathbb{N},
\end{equation}
and
\begin{equation}
\textbf{MPCR2$'$:} \,\, ~~~~~~~~~~W''_{n+1}(L)=\begin{cases}
c_A, &\text{if }W_n(L)\in B,\\
c_B, &\text{if }W_n(L)\in A,
\end{cases}
\end{equation}
where $c_A\in A$ and $c_B\in B$ are some  constants.

\item[{\bf Step 2:}]  {\bf The prediction of the order for $R^{(k+1)}$}, denoted by $O'(R^{(k+1)})$. More precisely, 
\begin{equation*}
\textbf{MPO1$'$:} \,\,~~~~~~~~ O'(R^{(k+1)})=\begin{cases}
O(\mbox{\mbox{Rev}}(R^{(l(k+1))})), &\text{if }W''_{n+1}(L)\in [\frac{1}{2},1]\\
O(R^{(k+1)}), &\text{if }W''_{n+1}(L)\in [0,\frac{1}{2}]
\end{cases}
\end{equation*}
and
\begin{equation*}
\textbf{MPO2$'$:} \,\, ~~~~~~~~~~~~~O'(R^{(k+1)})=\begin{cases}
O((R^{(l(l(k+1)))})), &\text{if }W''_{n+1}(L)\in [\frac{1}{2},1]\\
O(R^{l(k+1)}), &\text{if }W''_{n+1}(L)\in [0,\frac{1}{2}]\, .
\end{cases}
\end{equation*}
\item[{\bf Step 3:}] {\bf The prediction of $R^{(k+1)}$}, denoted by $R'^{(k+1)}$. 
For MPO1$'$,  
\begin{equation*}
R'^{(k+1)}=\begin{cases}
\mbox{Rev}(R^{(l(k+1))}), &\text{if }W''_{n+1}(L)\in [\frac{1}{2},1]\\
R^{l(k+1)}, &\text{if }W''_{n+1}(L)\in [0,\frac{1}{2}]
\end{cases}
\end{equation*}
and, for MPO2$'$,  
\begin{equation*}
R'^{(k+1)}=\begin{cases}
R^{(l(l(k+1)))}, &\text{if }W''_{n+1}(L)\in [\frac{1}{2},1]\\
R^{(l(k+1))}, &\text{if }W''_{n+1}(L)\in [0,\frac{1}{2}].
\end{cases}
\end{equation*}
\end{itemize}

\begin{rem}
{\rm There are the other alternatives to define the order of the return matrix $R^{(k)}$. Assume that there are three currencies in the foreign exchange market.  
Define the following counting measure
\begin{equation}\label{on}
\begin{split}
\pi^{k}&=\sharp\Big\{(i,j)\in \Lambda_3\times\Lambda_3: \bar r_{ij}^*\in\Big[ \max_{v\in\Lambda_3}\{r_{lv}^{(k)}\}-\varepsilon,
~~\max_{l,v\in\Lambda_3}\{r_{lv}^{(k)}\} \Big] \\
&~~~~~~\mbox{ or } \underline r_{ij}^*\in\Big[
\max_{l,v\in\Lambda_3}\{r_{lv}^{(k)}\}-\varepsilon,
~~\max_{l,v\in\Lambda_3}\{r_{lv}^{(k)}\} \Big]  \Big\},
\end{split}
\end{equation}
where $\Lambda_3:=\{1,2,3\}$ and $\varepsilon>0$ is some constant. It is easy to see that $\pi^{k}\in \{1,2,3\}$. In the sequel, let $R^{(k)}$ be the return matrix of the currency $1$, the currency $2$ and the currency $3$ at the day $k$, which admits the form below
\begin{equation}\label{example}
R^{(k)}=\left(\begin{array}{ccccc}
        0 & \bar r_{12}^{(k)}& \bar r_{13}^{(k)}\\
        0 & 0&0\\
        0&\underline  r_{32}^{(k)}&0\\
        \end{array}
        \right),
\end{equation}
where $\bar r_{12}^{(k)}\neq0, \bar r_{13}^{(k)}\neq0$ and $\underline r_{32}^{(k)}\neq0$. Next, we define the order $\tilde
O(R^{(k)})$ of $R^{(k)}$ by
\begin{equation}\label{Duonao}
\tilde O\Big(R^{(k)}\Big)=
\begin{cases}
(1,2),~~~~~~~~~~~~~~~~~\bar r_{12}^{(k)}=\max\{\bar r_{12}^{(k)}, \bar r_{13}^{(k)},\underline r_{32}^{(k)}\},\\
(1,3),~~~~~~~~~~~~~~~~~\bar r_{13}^{(k)}=\max\{\bar r_{12}^{(k)}, \bar r_{13}^{(k)},\underline r_{32}^{(k)}\},\\
(3,2),~~~~~~~~~~~~~~~~~\bar r_{32}^{(k)}=\max\{\bar r_{12}^{(k)}, \bar r_{13}^{(k)},\underline r_{32}^{(k)}\},\\
(1,2)\cup(1,3),~~~~~~~\bar r_{12}^{(k)}=\bar r_{13}^{(k)}=\max\{\bar r_{12}^{(k)}, \bar r_{13}^{(k)},\underline r_{32}^{(k)}\},\\
(1,3)\cup(3,2),~~~~~~~\bar r_{13}^{(k)}=\underline r_{32}^{(k)}=\max\{\bar r_{12}^{(k)}, \bar r_{13}^{(k)},\underline r_{32}^{(k)}\},\\
(3,2)\cup (1,2),~~~~~~~\bar r_{32}^{(k)}=\bar r_{12}^{(k)}=\max\{\bar r_{12}^{(k)}, \bar r_{13}^{(k)},\underline r_{32}^{(k)}\},\\
\Delta,~~~~~~~~~~~~~~~~~~~~~\bar r_{32}^{(k)}=\bar r_{12}^{(k)}=\bar
r_{13}^{(k)}=\max\{\bar r_{12}^{(k)}, \bar r_{13}^{(k)},\underline
r_{32}^{(k)}\}.
\end{cases}
 \end{equation}
While, the order above does not work very well to show the effectiveness of the cross rate method. By the cluster idea, for the first three case (i.e., $\pi^{k}=1$ ), we can regard the pairs  $(1,2)$, $1,3)$ and $(3,2)$ as the same. Also, for the cases 4-6 (i.e., $\pi^{k}=2$), we  regard the pairs $(1,2)\cup(1,3)$, $(1,3)\cup(3,2)$ and $(3,2)\cup (1,2)$ are identical. So, we modify the order $\tilde O(R^{(k)})$  to redefine the order $O(R^{(k)})$ of $R^{(k)}$ by
\begin{equation}\label{ce3}
O(R^{(k)})=
\begin{cases}
1,~~~~~~~~~\pi^{k}=1,\\
2,~~~~~~~~~\pi^{k}=2,\\
0,~~~~~~~~~\pi^{k}=3.
\end{cases}
\end{equation}
Although the order \eqref{ce3} works for the first two steps of the cross rate method, it is unavailable to predict the value of
$R^{(k)}$ in the third step since we cannot write explicitly the reverse of $R^{(k)}$.}
\end{rem}
 
\section{Main results}
 
\subsection{The cross rate scheme} 
Set 
 \begin{equation}\label{Futon}
\theta_n(CR):=\frac{\sharp\{O'(R^{(k)}):O'(R^{(k)})=O(R^{(k)}), R^{(k)}\in \acute{D}_n(L))\}}{L}.
\end{equation}
The numerator on the right hand side of \eqref{Futon} counts the total number that  the prediction order is the same as the genuine order of $R^{(k)}$ during the trading day from the day $(n-1)L+1$ to the day $nL.$ In \eqref{Futon}, $\theta_n$ is called 
 the success rate of the CR(MPCR, MPO, $R'^{(K+1)}$) for the segment sequence $ \acute{D}_n(L)$. 
If
\begin{equation*}
\theta_n(CR)\geqslant \frac{1}{2},
\end{equation*}
then we say CR(MPCR, MPO, $R'^{(K+1)}$) is \textit{effective} for the segment $ \acute{D}_n(L)$ . 

Using the two update rules  IITC (\eqref{portf}) or EIITC(\eqref{portf1}) with the effective CR(MPCR, MPO, $R'^{(K+1)}$), we define a profitable strategy for the whole daily return sequence $ \acute{D}_n(L)$, 
as follows 
\begin{align}\label{cr}
~&\eta (\acute{D}(N),MPCR, MPO,R^{(k+1)})\nonumber\\
:=&\frac{\sharp\{\acute{D}_n(L):CR(MPCR,MPO,R'^{(k+1)}) \text{is effective for \'{D}}_n(L)\}}{\sharp\{\acute{D}_n(L)\}}
\geqslant \frac{1}{2}.
\end{align}
In the trading day $k+1$, investors apply the result of effective  CR(MPCR,MPO,$R'^{(K+1)}$) in the $ \acute{D}_n(L)$ and combine the  two update rules ( IITC (\eqref{portf})  and  EIITC(\eqref{portf1}) ) to update their portfolios to gain more profits.
 
\begin{lemma}\label{eff}
 
The CR($MPCR^{a}$, $MPO^{b}$, $R'^{(K+1)}$) for a,b=1,2 with segment sequence $ \acute{D}_n(L)$ is effective if we  hold 
 either
\begin{equation}\label{cond1}
W'_{n+1}(L), W_{n+1}(L)\in [0,\frac{1}{2})
\end{equation}
or
\begin{equation}\label{cond2}
W'_{n+1}(L), W_{n+1}(L)\in [\frac{1}{2},1]. 
\end{equation}
\end{lemma}
 
\begin{proof}
If $W_{n+1}(L)\in[0,\frac{1}{2})$, then more than half  points of the set $\Theta:=\{nL+k,k=1,2,\cdots, L\}$ are not cross positions. In what follows, we take $nL+k+1\in \Theta$ and assume that $nL+k+1$ is not a cross position so that 
\begin{equation}\label{Talbot}
O(R^{(nL+k+1)})=O(R^{(nL+k)}).
\end{equation}
Next, due to $W_{n+1}(L)\in[0,\frac{1}{2})$, we have 
\begin{equation}\label{Faraday}
O'(R^{(nL+k+1)})=O(R^{(nL+k)}) 
\end{equation}
according to MPO1. Therefore, \eqref{Talbot} and \eqref{Faraday} yields that 
\begin{equation} \label{Tenby}
O(R^{(nL+k+1)})=O'(R^{(nL+k+1)}).
\end{equation}
Therefore, CR$(\mbox{MPCR1,  MPO1},~  R'^{(k+1)})$ is effective whenever  $W'_{n+1}(L), W_{n+1}(L)\in [0,\frac{1}{2})$.

If $W_{n+1}(L)\in[\frac{1}{2},1]$, then more than half  points of the set $\Theta$  are cross positions. In the sequel, we take 
$nL+k+1\in \Theta$ and assume that $nL+k+1$ is   a cross position such that 
\begin{equation}\label{swanseabay}
O(R^{(nL+k+1)})\neq O(R^{(nL+k)}).
\end{equation}
On the other hand, if $W_{n+1}(L)\in[\frac{1}{2},1]$, thus one has 
\begin{equation}\label{Oxwich}
O'(R^{(nL+k+1)})=O(Rev(R^{(nL+k)})).
\end{equation}
If $O(R^{(nL+k)})=1$, then $O'(R^{(nL+k+1)})=2$ by \eqref{Oxwich}. Moreover, from \eqref{swanseabay}, it follow that $O(R^{(nL+k+1)})=2$ by noting that $O(\cdot)$ takes only two values. Therefore,   \eqref{Tenby} holds. 
 Likewise, if  $O(R^{(nL+k)})=1$, we can deduce that \eqref{Tenby} is true. In all, \eqref{Tenby}  holds true for any cases. Consequently, CR$(\mbox{MPCR1,  MPO1},~  R'^{(k+1)})$ is effective provided that $W'_{n+1}(L), W_{n+1}(L)\in [\frac{1}{2},1]$.  
 
 Below, we assume that $nL+k+1\in \Theta$ such that \eqref{swanseabay}. In case $W_{n+1}'(L)\in[\frac{1}{2},1]$, in the light of MPO2, one has 
 \begin{equation}\label{airshow}
O'(R^{(nL+k+1)})=O(R^{(nL+k-1)}).
\end{equation}
If $O(R^{(nL+k-1)})=O(R^{(nL+k)})$, then we can deduce that more than half points of the set $\Theta$, which is contradictory with $W_{n+1}(L)\in[\frac{1}{2},1]$. As a consequence, we arrive at 
\begin{equation}\label{Iphone}
 O(R^{(nL+k-1)})\neq O(R^{(nL+k)})
 \end{equation}
 Taking \eqref{swanseabay}, \eqref{airshow} as well as \eqref{Iphone} into consideration, we derive that 
 \begin{equation}\label{wechat}
O'(R^{(nL+k+1)})=O(R^{(nL+k+1)}).
\end{equation}
Indeed, if $O(R^{(nL+k)})=1$, then $O(R^{(nL+k+1)})=1$ from \eqref{swanseabay} and 
 $O(R^{(nL+k-1)})=2$ due to \eqref{Iphone}, which implies $O'(R^{(nL+k+1)})=1$. In a similar way, we can show that \eqref{wechat} holds true. Therefore, CR$(\mbox{MPCR1,  MPO2},~  R'^{(k+1)})$ is effective for $W'_{n+1}(L), W_{n+1}(L)\in [\frac{1}{2},1]$.   Since the other situations can be dealt with similarly, we herein omit the corresponding details. 
\end{proof}

By a close inspection of the lemma above,  we deduce  that CR($MPCR^{a}$,  $MPO^{b}$, $R'^{(k+1)}$) for $ a,b=1,2 $ with the  segment  $ \acute{D}_n(L)$ is effective in the case that  both $ W_{n+1}(L)$ and  $ W'_{n+1}(L)$ belong   to the same intervals $[0,\frac{1}{2})$  or 
$[\frac{1}{2},1]$. Nevertheless,    the values of $ W_{n+1}(L)$ and  $ W'_{n+1}(L)$  need not to be identical. 
 
 \begin{defn}\label{London}
 A sequence  $\{\varsigma_n\}_{n\ge 1}$  is called finitely dependent if there exists some $K>0$ such that $\varsigma_n\ $ and $\varsigma_{n+K1}\ $ are independent for any $n$, $K_1\in \mathbb{N}$ and $K_1\geqslant K$. 
 \end{defn}
  
In particular, by taking $K_1=K$, Definition \ref{London} shows that $\varsigma_n\ $is independent of $\varsigma_{n+K_1}\ $, but need not to be independent of $\varsigma_k$ for  $n<k<n+K_1$.
 
The following lemma is taken from \cite{ALZ2001}. 
 
\begin{lemma}\label{lm5.2}
If a sequence $\{\varsigma_n\}_{n\geqslant1}$ is finitely dependent of bounded random variables and $\mathbb{E}\varsigma _n\geqslant c$, for some constant $c$ and for any $n\in \mathbb{N}$, then
\begin{equation}
\lim_{N\rightarrow\infty} \frac{1}{N}\sum_{n=1}^{N} \varsigma_n\geqslant c, \quad a.s.
\end{equation}
\end{lemma}

Based on this lemma, the profitability of the IITC and the EIITC can be obtained. We state the following

\begin{theorem}
We assume that  $\{R_n(L)\}$ is finitely dependent sequence of cross rate, then we have following two result 
 
(1) if
\begin{equation}\label{thmcond1}
P_{AA}+P_{BB}\geq\frac{1}{2},
\end{equation}
then two update rules IITC or EIITC with $ CR(MPCR^a, MPO^b,R'^{(k+1)}),$ $ a,b=1,2$, become a profitable strategy when time horizon goes to infinity.
 
(2)    if
\begin{equation} \label{thmcond2}
P_{AB}+P_{BA}\geq\frac{1}{2},
\end{equation}
then two update rules IITC or EIITC with $ CR(MPCR^a, MPO^b,R'^{(k+1)}),$ $ a,b=1,2$, become a profitable strategy when time horizon goes to infinity.
\end{theorem}

\begin{proof}
We start with the proof for the case (1) . Let
\begin{equation}
\varsigma_n=
\begin{cases}
1, & \text{if } \theta_n(MPCR1,MPO2,R'^{(k+1)}) \geq \frac{1}{2} \\
0,  & \text{if } \theta_n(MPCR1,MPO2,R'^{(k+1)}) < \frac{1}{2}
\end{cases}
\end{equation}
for $n=1,2,\cdots$. By the definition of expectation for discrete time random variable, $E\varsigma_n\ge\frac{1}{2}$.
Note that 
\begin{equation*}
\begin{split}
&\{W'_{n+1}(L),W_{n+1}(L)\in A\}\cup \{W'_{n+1}(L),W_{n+1}(L)\in B\}\\
&\subseteq\{CR(MPCR1, MPO1,R'^{(K+1)})  \mbox{ is effective }\}\\
&=\{\varsigma_n=1\}.
\end{split}
\end{equation*}
So, one has 
\begin{equation}\label{zuimei}
\begin{split}
&\mathbb{P}(\{W'_{n+1}(L),W_{n+1}(L)\in A\}\cup \{W'_{n+1}(L),W_{n+1}(L)\in B\})
\le \mathbb{P}(\varsigma_n=1).
\end{split}
\end{equation}
Due to the fact that 
\begin{equation*}
\begin{split}
&\mathbb{P}(\{W'_{n+1}(L),W_{n+1}(L)\in A\}\cup \{W'_{n+1}(L),W_{n+1}(L)\in B\})
=P_{AA}+P_{BB},
\end{split}
\end{equation*}
together with the assumption \eqref{thmcond1}, we deduce from \eqref{zuimei} that 
\begin{equation*}
\mathbb{P}(\varsigma_n=1)\geq P_{AA}+P_{BB}\geq\frac{1}{2}.
\end{equation*}
Next, applying Lemma \ref{lm5.2} yields that 
\begin{equation*}
\lim_{N\rightarrow \infty} \eta (\\\acute{D}(N),MPCR,MPO,R'^{(k+1)})=\lim_{n\rightarrow \infty}\frac{1}{n} \sum_{k=1}^n \varsigma_k\geq \frac{1}{2}, \quad a.s.
\end{equation*}
This completes the proof.

Next, we move forward to complete the proof concerning the case (2).  Set \begin{equation*}
\varsigma_n=
\begin{cases}
1, & \text{if } \theta_n(MPCR2,MPO1,R'^{(k+1)}) \geq \frac{1}{2} \\
0,  & \text{if } \theta_n(MPCR2,MPO1,R'^{(k+1)}) < \frac{1}{2}
\end{cases}
\end{equation*}
for $n=1,2,\cdots$.  It is easy to see that $E\varsigma_n\ge\frac{1}{2}$. Since
\begin{equation*}
\begin{split}
&\{W'_{n+1}(L)\in A,W_{n+1}(L)\in B\}\cup \{W'_{n+1}(L)\in B,W_{n+1}(L)\in A\}\\
&\subseteq\{CR(MPCR2, MPO1,R'^{(K+1)})  \mbox{ is effective }\}\\
&=\{\varsigma_n=1\},
\end{split}
\end{equation*}
we have 
\begin{equation*} 
\begin{split}
&\mathbb{P}(\{W'_{n+1}(L)\in A,W_{n+1}(L)\in B\}\cup \{W'_{n+1}(L)\in B,W_{n+1}(L)\in A\})
\le \mathbb{P}(\varsigma_n=1).
\end{split}
\end{equation*}
This, together with 
\begin{equation*}
\begin{split}
&\mathbb{P}(\{W'_{n+1}(L),R_{n+1}(L)\in A\}\cup \{W'_{n+1}(L),W_{n+1}(L)\in B\})
=P_{AB}+P_{BA},
\end{split}
\end{equation*}
and  \eqref{thmcond2}, leads to 
\begin{equation*}
\mathbb{P}(\varsigma_n=1)\geq P_{AB}+P_{BA}\geq\frac{1}{2}.
\end{equation*}
Thus, the desired assertion follows from Lemma \ref{lm5.2}. 
\end{proof}
 
The key points  for the selections of the MPCR1 and the MPCR2 are based on the theorem above.
 
\begin{remark}
Above theorem shows the general situation for the profitable portfolio selection. In the real world, investors can select  pairs of currencies in the foreign exchange market with one of $P_{AA}+P_{BB},P_{AB}+P_{BA},P_{AA}+P_{BA}$ and $P_{AB}+P_{BB}$ which with the value greater than $\frac{1}{2}$. For instance, if investors select five pairs of currencies by $P_{AA}+P_{BB}=0.78$, then we have a profitable strategy,
\begin{equation}
\lim_{\eta \rightarrow\infty}\pi(\text{\'{D}}(N),MPCR,MPO,R'^{(k+1)})\geq 0.78, \text{ } a.s..\end{equation}
\end{remark}
 
\subsection{Universality of the IITC and the EIITC}
 Motivated by \cite{ALZ2001},  in this section we aim  to show the universality of the on-line portfolio selections (i.e., \eqref{portf} and \eqref{portf1}) in the foreign exchange markets. Clearly, two update rules (i.e.,  IITC and EIITC)  are universality for both active and passive strategies. 

For simplicity, we just take a single pair $(i,j) $ of the currencies involved.  With the help \eqref{b-p},   the exponential growth rate of investment on the currency pair $(i,j) $    is
\begin{equation}\label{ren1}
LI^*_N(e_{ij},\{R^{(k)}\})=\frac{1}{N}\log\bigg(\prod_{k+1}^N e_{ij}\diamond R^{(k)}\bigg)=\frac{1}{N}\sum_{k=1}^{N}\log (e_{ij}\diamond R^{(k)}),
\end{equation}
where $e_{ij} \in \mathbb{R}^m\otimes \mathbb{R}^m$, where  the $ij^{th}$ entry is equal to $1$ and the other entries are equal to zero. Recall from   \eqref{P-B} that the  exponential growth rate $R_N(\{\psi^{(k)}\},\{R^{(k)}\})$ with the transaction costs is defined as  
\begin{equation}\label{Ren}
R_N(\{\psi^{(k)}\},\{R^{(k)}\})\:= \frac{1}{N}\sum_{i=1}^{N}\log \Big(\psi^{(k)}\diamond R^{(k)}\Big)+\frac{1}{N}\sum_{i=1}^{N}\log (1-c_k).\end{equation} 

The following theorem reveals  the gap between $R_N(\{\psi^{(k)}\},\{R^{(k)}\})$ and $LI^*_N(e_{ij},\{R^{(k)}\})$. 

\begin{theorem}\label{R}
Let $R^{(1)}, \cdots, R^{(N)}$ be an arbitrary sequence of return matrices with $\bar r_{ij}^{(k)}+\underline r_{ij}^{(k)}\ge r,$ where $ i,j\in\Lambda_0,k\in\Lambda_1$,  for some constant $r\in(0,1)$  and $ 
\max_{i,j\in\Lambda_0}(\bar r_{ij}^{(k)}+\underline r_{ij}^{(k)})=1 $. Consider the linear prediction ${R'}^{(k+1)}=\sum_{l=1}^{d_k}a_{k,l}R^{(k-l+1)}$, where $a_{k,l}\ge0,~ l=1,\cdots, d_k$ and                   
 $~\sum_{l=1}^{d_k}a_{k,l}=1, d_k\ge1.$  Let $\gamma>0.$ Then,  
 
 \begin{itemize}
\item For the IITC algorithm \eqref{portf},
\begin{equation} \label{ren9}
\begin{split}
&R_N(\{\psi^{(k)}\},\{R^{(k)}\})-R_N^*(\{e_{ij}\},\{R^{(k)}\})\\
&\ge\frac{1}{N}\log\frac{\psi^{(1)}_{ij}}{\psi^{(N+1)}_{ij}}+\frac{1}{N}\sum_{k=1}^{N}\log (1-c_k)+\gamma r-\gamma
.
\end{split}
\end{equation}
\item For  the EIITC algorithm \eqref{portf1},
\begin{equation} \label{ren10}
\begin{split}
&R_N(\{\psi^{(k)}\},\{R^{(k)}\})-R_N^*(\{e_{ij}\},\{R^{(k)}\})\\
&\ge\frac{1}{N}\log\frac{\psi^{(1)}_{ij}}{\psi^{(N+1)}_{ij}}+\frac{1}{N}\sum_{k=1}^{N}\log (1-c_k)+\gamma r-\frac{\gamma}{r}
.
\end{split}
\end{equation}
\item For the IITC algorithm \eqref{portf},
\begin{equation}\label{PP}
c_{k+1}(\gamma)\leqslant \frac{c\e^\gamma}\gamma{1-c}+O(\gamma^2),
\end{equation}
where $O(\gamma^2)$ means that there exist  $\alpha_1,\alpha_2\in R$  such that $\alpha_1\gamma^2\le O(\gamma^2)\le \alpha_2\gamma^2$ as $\gamma>0$ is sufficiently small.
\item For  the EIITC algorithm \eqref{portf1}, 
\begin{equation}\label{PPP}
c_{k+1}(\gamma)\leqslant \frac{c\e^{\gamma/r}\gamma}{(1-c)r}+O(\gamma^2).
\end{equation}
\end{itemize}

\end{theorem}

\begin{proof}
We only focus on the proof of \eqref{ren9} since \eqref{ren10}  can be done in a similar manner. 
It is easy to see that 
\begin{equation*}
\log\frac{\psi^{(1)}_{ij}}{\psi^{(N+1)}_{ij}}=\sum_{k=1}^N\Big\{\log\Big(\psi^{(k)}_{ij}\Big)-\log\Big(\psi^{(k+1)}_{ij}\Big)\Big\}.
\end{equation*}
This, together with \eqref{portf}, yields that 
\begin{equation}\label{ren4}
\log\frac{\psi^{(1)}_{ij}}{\psi^{(N+1)}_{ij}}=\sum_{k=1}^N\bigg\{\log\Big(\psi^{(k)}_{ij}\Big)-\log\bigg(\frac{\psi'^{(k)}_{ij}\exp(\gamma (\bar r_{ij}^{'(k+1)}+\underline r_{ij}^{'(k+1)}))}{\sum_{v,l=1}^{m}\psi'^{(k)}_{vl}\exp(\gamma (\bar r_{vl}^{'(k+1)}+\underline r_{vl}^{'(k+1)}))}\bigg)\bigg\}.
\end{equation}
Due to  \eqref{p1}, we have 
$$\psi^{(k)}_{ij}=\frac{{\psi'}^{(k)}_{ij}(\psi^{(k)}\diamond R^{(k)})}{\bar r^{(k)}_{ij}+\underline r^{(k)}_{ij}}.$$
Substituting this into \eqref{ren4} gives that 
\begin{equation} \label{ren5}
\begin{split}
\log\frac{\psi^{(1)}_{ij}}{\psi^{(N+1)}_{ij}}&=\sum_{k=1}^N\bigg\{\log\bigg(\frac{{\psi'}^{(k)}_{ij}(\psi^{(k)}\diamond R^{(k)})}{\bar r^{(k)}_{ij}+\underline r^{(k)}_{ij}}\bigg)\\
&\quad-\log\bigg(\frac{\psi'^{(k)}_{ij}\exp(\gamma (\bar r_{ij}^{'(k+1)}+\underline r_{ij}^{'(k+1)}))}{\sum_{v,l=1}^{m}\psi'^{(k)}_{vl}\exp(\gamma (\bar r_{vl}^{'(k+1)}+\underline r_{vl}^{'(k+1)}))}\bigg)\bigg\}\\
&=\sum_{k=1}^N\bigg\{\log\bigg({\psi'}^{(k)}_{ij}\bigg)+\log\bigg(\psi^{(k)}\diamond R^{(k)})\bigg)-\log\bigg(\bar r^{(k)}_{ij}+\underline r^{(k)}_{ij}\bigg)\\
&\quad-\log\bigg({\psi'}^{(k)}_{ij}\bigg)-\gamma (\bar r_{ij}^{'(k+1)}+\underline r_{ij}^{'(k+1)})\\
&\quad+\log\bigg(\sum_{v,l=1}^{m}\psi'^{(k)}_{vl}\exp(\gamma (\bar r_{vl}^{'(k+1)}+\underline r_{vl}^{'(k+1)}))\bigg)\bigg\}\\
&=\sum_{k=1}^N\bigg\{\log\bigg(\psi^{(k)}\diamond R^{(k)})\bigg)-\log\bigg(\bar r^{(k)}_{ij}+\underline r^{(k)}_{ij}\bigg)\\
&\quad-\gamma (\bar r_{ij}^{'(k+1)}+\underline r_{ij}^{'(k+1)})\\
&\quad+\log\bigg(\sum_{v,l=1}^{m}\psi'^{(k)}_{vl}\exp(\gamma (\bar r_{vl}^{'(k+1)}+\underline r_{vl}^{'(k+1)}))\bigg)\bigg\}.
\end{split}
\end{equation}
On the other hand, 
by taking  \eqref{ren1}  and \eqref{Ren}into account, it follows from \eqref{ren5} that 
\begin{equation}\label{ren8}
\begin{split}
&R_N(\{\psi^{(k)}\},\{R^{(k)}\})-R_N^*(\{e_{ij}\},\{R^{(k)}\})\\
&=\frac{1}{N}\sum_{k=1}^{N}\Big\{\log \bigg( \psi^{(k)}\diamond R^{(k)}\bigg)+\log (1-c_k)-\log \bigg(e_{ij}\diamond R^{(k)}\bigg)\Big\}\\
&=\frac{1}{N}\log\frac{\psi^{(1)}_{ij}}{\psi^{(N+1)}_{ij}}+\frac{1}{N}\sum_{k=1}^{N}\log (1-c_k)\\
&\quad+\frac{1}{N}\sum_{k=1}^{N}\bigg(\gamma(\bar r_{ij}^{'(k+1)}+\underline r_{ij}^{'(k+1)})-\log\bigg(\sum_{v,l=1}^{m}\psi'^{(k)}_{vl}\exp(\gamma (\bar r_{vl}^{'(k+1)}+\underline r_{vl}^{'(k+1)}))\bigg)\bigg).
\end{split}
\end{equation}
Note that 
\begin{equation*}
\bar r_{ij}^{'(k+1)}+\underline r_{ij}^{'(k+1)})=\sum_{l=1}^{d_k}a_{k,l}(\bar r_{ij}^{(k-l+1)}+\underline r_{ij}^{(k-l+1)})
\end{equation*}
This, together with $\max_{i,j}(\bar r_{ij}^{(k)}+\underline r_{ij}^{(k)})=1$ and $\bar r_{ij}^{(k)}+\underline r_{ij}^{(k)}\ge r$, leads to 
\begin{equation}\label{Budapest}
r\le \bar r_{ij}^{'(k+1)}+\underline r_{ij}^{'(k+1)})\le 1.
\end{equation}
Therefore, we conclude that 
\begin{equation}\label{ren6}
\frac{1}{N}\sum_{k=1}^{N}\bigg(\gamma(\bar r_{ij}^{'(k+1)}+\underline r_{ij}^{'(k+1)})\bigg)\ge\gamma r
\end{equation}
and that 
\begin{equation}\label{ren7}
\log\bigg(\sum_{v,l=1}^{m}\psi'^{(k)}_{vl}\exp(\gamma (\bar r_{vl}^{'(k+1)}+\underline r_{vl}^{'(k+1)}))\bigg)\bigg)\le\log\bigg(\sum_{v,l=1}^{m}\psi'^{(k)}_{vl}\exp(\gamma )\bigg)\bigg)=\gamma,
\end{equation}
where in the last display we have used the fact $\sum_{v,l=1}^{m}\psi'^{(k)}_{vl}=1$. Inserting \eqref{ren6} and \eqref{ren7} back into \eqref{ren8} implies that 
\begin{equation} 
\begin{split}
&R_N(\{\psi^{(k)}\},\{R^{(k)}\})-R_N^*(\{e_{ij}\},\{R^{(k)}\})\\
&\ge\frac{1}{N}\log\frac{\psi^{(1)}_{ij}}{\psi^{(N+1)}_{ij}}+\frac{1}{N}\sum_{k=1}^{N}\log (1-c_k)+\gamma r-\gamma
.
\end{split}
\end{equation}
So the desired assertion \eqref{ren9} follows immediately.

In the sequel, we work only on \eqref{PP} since \eqref{PPP} can be coped with in a parallel way. By \eqref{Hefei}, in addition to \eqref{ineq2}, it follows that 
\begin{equation*}
\begin{split}
c_{k+1}(\gamma)&=\frac{T_{k+1}}{F_k}\\
&\le \frac{c}{1-c}d(\psi^{(k+1)}, {\psi'}^{(k)})\\
&= \frac{c}{1-c}\sum_{ij=1}^{m}|\psi^{(k+1)}_{ij}-\psi'^{(k)}_{ij}|.
\end{split}
\end{equation*}
This, together with \eqref{portf},  implies that 
\begin{equation*}
\begin{split}
c_{k+1}(\gamma)&\le\frac{c}{1-c}\sum_{ij=1}^{m}\bigg|\frac{\psi'^{(k)}_{ij}\exp(\gamma (\bar r_{ij}^{'(k+1)}+\underline r_{ij}^{'(k+1)}))}{\sum_{v,l=1}^{m}\psi'^{(k)}_{vl}\exp(\gamma (\bar r_{vl}^{'(k+1)}+\underline r_{vl}^{'(k+1)}))}-\psi'^{(k)}_{ij}\bigg|\\
&=\frac{c}{1-c}\sum_{ij=1}^{m}\psi'^{(k)}_{ij}\bigg|\frac{\exp(\gamma (\bar r_{ij}^{'(k+1)}+\underline r_{ij}^{'(k+1)}))}{\sum_{v,l=1}^{m}\psi'^{(k)}_{vl}\exp(\gamma (\bar r_{vl}^{'(k+1)}+\underline r_{vl}^{'(k+1)}))}-1\bigg|\\
&=\frac{c}{1-c}\sum_{ij=1}^{m}\psi'^{(k)}_{ij}\bigg|\frac{\exp(\gamma (\bar r_{ij}^{'(k+1)}+\underline r_{ij}^{'(k+1)}))-\sum_{v,l=1}^{m}\psi'^{(k)}_{vl}\exp(\gamma (\bar r_{vl}^{'(k+1)}+\underline r_{vl}^{'(k+1)}))}{\sum_{v,l=1}^{m}\psi'^{(k)}_{vl}\exp(\gamma (\bar r_{vl}^{'(k+1)}+\underline r_{vl}^{'(k+1)}))}\bigg|\\
&=\frac{c}{1-c}\sum_{ij=1}^{m}\psi'^{(k)}_{ij}\bigg|\frac{\sum_{v,l=1}^{m}\psi'^{(k)}_{vl}  \exp(\gamma (\bar r_{ij}^{'(k+1)}+\underline r_{ij}^{'(k+1)}))-\sum_{v,l=1}^{m}\psi'^{(k)}_{vl}\exp(\gamma (\bar r_{vl}^{'(k+1)}+\underline r_{vl}^{'(k+1)}))}{\sum_{v,l=1}^{m}\psi'^{(k)}_{vl}\exp(\gamma (\bar r_{vl}^{'(k+1)}+\underline r_{vl}^{'(k+1)}))}\bigg|,
\end{split}
\end{equation*}
where in the last step we have used $\sum_{v,l=1}^{m}\psi'^{(k)}_{vl}=1$. Thus, we deduce from \eqref{Budapest} and $\sum_{v,l=1}^{m}\psi'^{(k)}_{vl}=1$ that 
 \begin{equation}\label{Way}
\begin{split}
c_{k+1}(\gamma)&\le\frac{c}{1-c}\sum_{ij=1}^{m}\psi'^{(k)}_{ij}\sum_{v,l=1}^{m}\psi'^{(k)}_{vl} \bigg|\frac{ \exp(\gamma (\bar r_{ij}^{'(k+1)}+\underline r_{ij}^{'(k+1)}))-\exp(\gamma (\bar r_{vl}^{'(k+1)}+\underline r_{vl}^{'(k+1)}))}{\sum_{v,l=1}^{m}\psi'^{(k)}_{vl}\exp(\gamma (\bar r_{vl}^{'(k+1)}+\underline r_{vl}^{'(k+1)}))}\bigg|\\
&\le \frac{c\e^{-\gamma r}}{1-c}\sum_{ij=1}^{m}\psi'^{(k)}_{ij}\sum_{v,l=1}^{m}\psi'^{(k)}_{vl} \bigg| \exp(\gamma (\bar r_{ij}^{'(k+1)}+\underline r_{ij}^{'(k+1)}))-\exp(\gamma (\bar r_{vl}^{'(k+1)}+\underline r_{vl}^{'(k+1)}))\bigg|\\
&\le \frac{c\e^{-\gamma r}}{1-c}\sum_{ij=1}^{m}\psi'^{(k)}_{ij}\sum_{v,l=1}^{m}\psi'^{(k)}_{vl} (\e^\gamma-\e^{\gamma r})\\
&=\frac{c\e^{-\gamma r}}{1-c} (\e^\gamma-\e^{\gamma r})\\
&=\frac{c\e^{\gamma (1-r)}}{1-c} (1-\e^{-\gamma(1-r)}).
\end{split}
\end{equation}
By the Taylor expansion, one has 
\begin{equation*}
\e^{-\gamma(1-r)}=1-\gamma(1-r)+O(\gamma^2).
\end{equation*}
Putting this into \eqref{Way} yields the desired assertion \eqref{PP}.
\end{proof}

The optimal currency  in the single trading action portfolio will be taken according to the IITC or the EIITC  update rule as the following strategy demonstrate.  More precisely, let us provide a partition of the set $\{1,...,N\}$ in the following manner:  for some integer $l\in(1,N)$, 
\begin{equation}
\Gamma_i=\begin{cases}
\{\frac{i(i-1)l}{2}+1,...,\frac{i(i+1)l}{2}\},\text i=1,2,...,n_l-1,\\
\{\frac{n_l(n_l-1)l}{2}+1,...,N\},\,\,\,\,\text  i=n_l,
\end{cases}
\end{equation}
in which $n_l:=\Big[\frac{\sqrt{1+8N/l}-1}{2}\Big]$ with $[x], x>0,$ being the smallest integer which is large or equal to $x$. It is readily to see that, for each $\Gamma_i, i<n_l,$  its  length  is equal to $il.$
 
In the sequel, we intend to show that, from a long-term point view of investment, the exponential growth rate of funds with decrements in terms of the IITC or the EIITC algorithm is bigger  than the one achieved by the single best currency. Let $\{R^{(i)}\}_{1\le i\le N}$ and $\{\psi^{(i)}\}_{1\leq i \leq N}$ be the sequences of return matrix and the portfolio matrix, respectively. Assume that  $\psi^{(i(i-1)l/2+l)},i=1,...,n,$ is bounded below by a small constant  $\varepsilon >0 $, and that  $\lambda(\cdot):i\mapsto\lambda(i)$ goes to zero as $i$ tends to infinity.  
 According to the partition of the set $\{1,\cdots, N\}$, it is easy to see that 
\begin{equation}\label{ren0}
\begin{split}
\sum_{k=1}^{N}\log (1-c_k)&=\log (1-c_1)+\log (1-c_2)+\cdots+\log (1-c_l)\\
&\quad+\log (1-c_{l+1})+\log (1-c_2)+\cdots+\log (1-c_{3l})\\
&\quad+\cdots\\
&\quad+\log (1-c_{\frac{(n_l-1)(n_l-2)}{2}+1})+\log (1-c_2)+\cdots+\log (1-c_{\frac{n_l(n_l-1)l}{2}})\\
&\quad+\log (1-c_{\frac{n_l(n_l-1)l}{2}}+1)+\cdots+\log (1-c_N)\\
&=\sum_{i=1}^{n_l-1}\sum_{k=1}^{il}\log(1-c_{(i(i-1)l)/2+k})+\sum_{k=1}^{L_{n_l}}\log(1-c_{(n_l(n_l-1)l)/2+k}).
\end{split}
\end{equation}
Thus, we deduce from \eqref{ren9} and \eqref{ren0} that 
\begin{equation} \label{re}
\begin{split}
&\liminf_{N\rightarrow\infty} R_N(\{\psi^{(k)}\},\{R^{(k)}\})-R_N^*(\{e_{ij}\},\{R^{(k)}\})\\
&\ge\liminf_{N\rightarrow\infty} \frac{1}{N}\Big(\sum_{i=1}^{n_l}\Big(\log\varepsilon+il\gamma(i)r-il\gamma(i)\\
&\quad+\sum_{i=1}^{n_l-1}\sum_{k=1}^{il}\log(1-c_{(i(i-1)l)/2+k})+\sum_{k=1}^{L_{n_l}}\log(1-c_{(n_l(n_l-1)l)/2+k})\Big) \\
&=0,
\end{split}
\end{equation}
where in the last procedure we have also used Theorem \ref{R} and the fact that \begin{equation*}\lim_{N\rightarrow \infty}\frac{n_l}{N}=\lim_{N\rightarrow \infty}\frac{L_{n_l}}{N}=0.\end{equation*}
  
From \eqref{re}, we can derive the following corollary, which state that 
\begin{cor}
From a long-term point view of investment, the exponential growth rate of funds with transaction cost in terms of the IITC or the EIITC algorithm is optimal than  the one achieved by the single 
best currency.
\end{cor}
\begin{remark}
In the case $\gamma =0$, there will be no action with the investors or the confidence of prediction is much  lower. Base on this,  the investor implements the buy-and-hold passive strategy.  
On the other hand, concerning the IITC or the EIITC algorithms, by virtue of \eqref{ren9} and \eqref{ren10},  we infer that the exponential growth rate of funds follow lower bounds. Most importantly, 
these algorithms show that investors will gain more whenever $\gamma \neq0$ with contrast to the case  $\gamma =0$.
\end{remark}

\section{Conclusions and further work}
We introduce a matrix-valued time series model for foreign exchange market according to the realistic market mechanism. Our construction captures the feature of the real foreign currency exchange 
markets in which the return matrix plays a key role. We are then able to define an order for the return matrix by looking at the unique maximum value, if it exists, either in the upper triangular part or 
lower triangular part of the return matrix. From this breakthrough point, we develop a cross rate method to establish an on-line portfolio selection scheme. Mathematically, we justify the profitability 
and the universality of constructed algorithm. 

In our paper, to define the order for two return matrices, we have eliminated the situation that there are more than one maximum value appeared either in the upper triangular part or the lower triangular 
part of the two return matrices, or even the more complex situation that the maximum value appeared in the both upper and lower triangular parts. This is more probably but remains a challenge 
mathematically. We would like also to mention that we have not yet to test our scheme developed in this paper with existing data from the currency exchange markets. We plan to consider these in 
our future work.

\end{document}